\documentclass[11pt]{article}
\usepackage{graphicx}
\usepackage{fullpage}
\usepackage{amsmath,amsthm,amssymb}
\usepackage{color}
\usepackage{verbatim}
\usepackage{mathabx}
\usepackage{algorithm}
\usepackage{algpseudocode}

\begin{document}
\def\lf{\left\lfloor}   
\def\rf{\right\rfloor}
\def\lc{\left\lceil}   
\def\rc{\right\rceil}

\newtheorem{theorem}{Theorem}[section]
\newtheorem{lemma}[theorem]{Lemma}
\newtheorem{proposition}[theorem]{Proposition}
\newtheorem{claim}[theorem]{Claim}
\newtheorem{corollary}[theorem]{Corollary}
\newtheorem{definition}[theorem]{Definition}
\newtheorem{observation}[theorem]{Observation}
\newtheorem{fact}[theorem]{Fact}
\newtheorem{property}{Property}
\newtheorem{remark}{Remark}[section]
\newtheorem{notation}{Notation}[section]
\newtheorem{example}{Example}[section]
\newtheorem{conjecture}{Conjecture}
\newtheorem{question}[conjecture]{Question}

\newcommand{\eps}{\varepsilon}
\newcommand{\tw}{\mathsf{tw}}
\newcommand{\pw}{\mathsf{pw}}
\newcommand{\tail}{\mathsf{tail}}
\newcommand{\diam}{\mathsf{diam}}
\newcommand{\dist}{\mathsf{dist}}
\newcommand{\ball}{\mathsf{ball}}
\newcommand{\OPT}{\mathsf{OPT}}
\newcommand{\dem}{\mathsf{dem}}
\newcommand{\lef}{\mathsf{left}}
\newcommand{\ri}{\mathsf{right}}
\newcommand{\XXX}{\textcolor{red}{XXX}}
\newcommand{\ImSucc}{S}
\newcommand{\Succ}{S'}
\newcommand{\len}{\mathsf{len}}
\newcommand{\supp}{\mathsf{supp}}
\newcommand{\UN}{\mathsf{UN}}
\newcommand{\nptime}{\textsf{NP }}
\newcommand{\dtime}{\textsf{DTIME }}
\newcommand{\zpptime}{\textsf{ZPP }}
\newcommand{\CC}[1]{\textcolor{red}{c_{#1}}}
\newcommand{\etal}{\textit{et al}.}
\newcommand{\poly}{\textit{poly}}
\newcommand{\spcut}{\text{Sparsest Cut }}
\newcommand{\fcg}{\text{flow-cut gap }}
\newcommand{\mcut}{\text{Multicut }}
\newcommand{\mflow}{\mathsf{maxflow }}
\newcommand{\gap}{\mathsf{gap }}

\title{On constant multi-commodity flow-cut gaps\\ for directed minor-free graphs}

\author{
Ario Salmasi\thanks{Dept.~of Computer Science \& Engineering, The Ohio State University, \texttt{salmasi.1@osu.edu}.}
\and
Anastasios Sidiropoulos\thanks{Dept.~of Computer Science, University of Illinois at Chicago, \texttt{sidiropo@uic.edu}.}
\and
Vijay Sridhar\thanks{Dept.~of Computer Science \& Engineering, The Ohio State University, \texttt{sridhar.38@osu.edu}.}
}

\clearpage

\date{}
\maketitle

\thispagestyle{empty}

\begin{abstract}
The multi-commodity flow-cut gap is a fundamental parameter that affects the performance of several divide \& conquer algorithms,
and has been extensively studied for various classes of undirected graphs.
It has been shown by Linial, London and Rabinovich \cite{linial1994geometry} and by Aumann and Rabani \cite{aumann1998log} that for general $n$-vertex graphs it is bounded by $O(\log n)$ and the Gupta-Newman-Rabinovich-Sinclair conjecture \cite{gupta2004cuts} asserts that it is $O(1)$ for any family of graphs that excludes some fixed minor.

The flow-cut gap is poorly understood for the case of directed graphs.
We show that for uniform demands it is $O(1)$ on directed series-parallel graphs, and on directed graphs of bounded pathwidth.
These are the first constant upper bounds of this type for some non-trivial family of directed graphs.
We also obtain $O(1)$ upper bounds for the general multi-commodity flow-cut gap on directed trees and cycles.
These bounds are obtained via new embeddings and Lipschitz quasipartitions for quasimetric spaces, which generalize analogous results form the metric case, and could be of independent interest.
Finally, we discuss limitations of methods that were developed for undirected graphs, such as random partitions, and random embeddings.
\end{abstract}

\pagebreak{}
\setcounter{page}{1}

\section{Introduction}
The multi-commodity flow-cut gap is a fundamental parameter that has been proven instrumental in the design of routing and divide \& conquer algorithms in graphs.
Bounds on this parameter generalize the max-flow/min-cut theorem, and lead to deep connections between algorithm design, graph theory, and geometry \cite{linial1994geometry,aumann1998log,arora2008euclidean}.
While the flow-cut gap for several classes of undirected graphs has been studied extensively, the case of directed graphs is poorly understood despite significant efforts. In this work we make progress towards overcoming this limitation by showing constant flow-cut gaps for some directed graph families. Consequently we also develop constant-factor approximation algorithms for certain directed cut problems on these graphs.

\subsection{Multi-commodity flow-cut gaps}
A \emph{multi-commodity flow} instance in an undirected graph $G$ is defined by two non-negative functions: $c \colon E(G) \to \mathbb{R}$ and $d \colon V(G) \times V(G) \to \mathbb{R}$. We refer to $c$ and $d$ as the \emph{capacity} and \emph{demand} functions respectively. The \emph{maximum concurrent flow} is the maximal value $\varepsilon$ such that for every $u,v \in V(G)$, $\varepsilon\cdot d(u,v)$ can be simultaneously routed between $u$ and $v$, without violating the edge capacities. We refer to this value as $\mflow(G,c,d)$.

For every $S \subseteq V(G)$, the \emph{sparsity} of $S$ is defined as follows:
\[
\frac{\sum_{(u,v)\in E(G)} c(u,v) |\mathbf{1}_S(u) - \mathbf{1}_S(v)|}{\sum_{u,v\in V(G)} d(u,v) |\mathbf{1}_S(u) - \mathbf{1}_S(v)|},
\]
where $\mathbf{1}_S \colon V(G) \to \{0,1\}$ is the indicator variable for membership in $S$. The sparsity of a cut is a natural upper bound for $\mflow(G,c,d)$. The \emph{multi-commodity max-flow min-cut gap} for $G$, denoted by $\gap(G)$, is the maximum gap between the value of the flow and the upper bounds given by the sparsity formula, over all multi-commodity flow instances on $G$.
The flow-cut gap on undirected graphs has been studied extensively, and several upper and lower bounds have been obtained for various graph classes. The gap is referred to as the \emph{uniform} multi-commodity flow-cut gap for the special case where there is a unit demand between every pair of vertices. 
Leighton and Rao \cite{leighton1999multicommodity} showed that the uniform flow-cut gap is $\Theta(\log n)$ in undirected graphs.
Subsequently Lineal, London and Rabinovich \cite{linial1994geometry} showed that the non-uniform multi-commodity flow-cut gap for the \spcut problem with $k$ demand pairs is upper bounded by $O(\log k)$.
Besides these there are various studies of the flow-cut gap for specific graph families. 
A central conjecture posed by Gupta, Newman, Rabinovich, and Sinclair in \cite{gupta2004cuts} asserts the following.

\begin{conjecture}[GNRS Conjecture \cite{gupta2004cuts}] \label{conj:gnrs}
For every family of finite graphs ${\cal F}$, we have $\gap({\cal F}) = O(1)$ iff ${\cal F}$ forbids some minor.
\end{conjecture}

Conjecture 1 has been verified for the case of series-parallel graphs \cite{gupta2004cuts}, $O(1)$-outerplanar graphs \cite{chekuri2006embedding}, $O(1)$-pathwidth graphs \cite{lee2009geometry}, and for some special classes of planar metrics \cite{sidiropoulos2013non}.
For graphs excluding any fixed minor the flow-cut gap with $k$ terminal pairs is known to be $O(1)$ for uniform demands and $O(\log k)$ for arbitrary demands \cite{klein1993excluded}. 

For the case of directed graphs, the flow-cut gap is defined in terms of 
the Directed Non-Bipartite \spcut problem which is an asymmetric variant of the \spcut problem, and is defined as follows. Let $G$ be a directed graph and let $c : E(G) \to \mathbb{R}_{\geq 0}$ be a capacity function. Let $T = \{(s_1,t_1),(s_2,t_2),\ldots,(s_k,t_k)\}$ be a set of terminal pairs, where each terminal pair $(s_i,t_i)$ has a non-negative demand $\dem(i)$. A cut in $G$ is a subset of directed edges of $E(G)$. For a cut $S \subseteq E(G)$ in $G$, let $I_S$ be the set of all indices $i \in \{1,2,\ldots,k\}$ such that all paths from $s_i$ to $t_i$ have at least one edge in $S$. Let $D(S) = \sum_{i\in I_S} \dem(i)$ be the demand separated by $S$. Let $W(S) = \frac{C(S)}{D(S)}$ be the \emph{sparsity} of $S$. The goal is to find the cut with minimum sparsity. The LP relaxation of this problem corresponds to the dual of the LP formulation of the directed maximum concurrent flow problem, and the integrality gap of this LP relaxation is the directed multi-commodity flow-cut gap. Hajiyaghayi and R\"{a}cke \cite{aprxdirectedsparsestcut} showed an upper bound of $O(\sqrt{n})$ for the flow-cut gap. This upper bound on the gap has been further improved by Agarwal, Alon and Charikar to $\tilde{O}(n^{11/23})$ in \cite{agarwal2007improved}. For directed graphs of treewidth $t$, it has been shown that the gap is at most $t \log^{O(1)} n $ by M{\'e}moli, Sidiropoulos and Sridhar \cite{memoli2016quasimetric}. On the lower bound side Saks \etal~\cite{saks2004lower} showed that for general directed graphs the flow-cut gap is at least $k-\eps$, for any constant $\eps>0$, and for any $k=O(\log n/\log\log n)$. Chuzhoy and Khanna showed a $\tilde{\Omega}(n^{\frac{1}{7}})$ lower bound for the flow-cut gap in \cite{chuzhoy2009polynomial}.


A natural generalization of the GNRS Conjecture for directed graphs poses the question of whether the multi-commodity flow-cut gap is $O(1)$ for any family of minor free directed graphs. In this paper, we provide the first constant gaps for some non-trivial family of graphs. Throughout this paper, when we refer to a directed family of graphs we mean that it is obtained from an undirected family of graphs by assigning arbitrary directions to the edge sets. We state below our two main results pertaining to the flow-cut gap.

\begin{theorem}\label{thm:uniformgap}
The uniform multi-commodity flow-cut gap on directed series-parallel graphs and directed bounded pathwidth graphs is $O(1)$.  
\end{theorem}

\begin{theorem}\label{thm:nonuniformgap}
The non-uniform multi-commodity flow-cut gap on directed cycles and directed trees is $O(1)$.  
\end{theorem}

\subsection{Cut problems of directed graphs}
Better bounds on the flow-cut gap typically also imply better approximation ratios for solving cut problems. For the Directed Non-Bipartite \spcut problem the flow-cut gap upper bounds of \cite{aprxdirectedsparsestcut} and \cite{agarwal2007improved} are also accompanied by $O(\sqrt{n})$ and $O(\tilde{n}^{11/23})$ polynomial time approximation algorithms respectively. Similarly for graphs of treewidth $t$, a $t \log^{O(1)} n $ polynomial time approximation algorithm is also provided in \cite{memoli2016quasimetric}.

Another closely related cut problem is the Directed \mcut problem which is defined as follows. Let $G$ be a directed graph and let $c : E(G) \to \mathbb{R}_{\geq 0}$ be a capacity function. Let $T = \{(s_1,t_1),(s_2,t_2),\ldots,(s_k,t_k)\}$ be a set of terminal pairs. A \emph{cut} in $G$ is a subset of $E(G)$. The \emph{capacity} of a cut $S$ is $c(S) = \sum_{e \in S}c(e)$. The goal is to find a cut separating all terminal pairs, minimizing the capacity of the cut. This problem is NP-hard. An $O(\sqrt{n \log{n}})$ approximation algorithm for Directed \mcut was presented by Cheriyan, Karloff and Rabani \cite{cheriyan2005approximating}. Subsequently an $\tilde{O}(n^{2/3}/\OPT^{1/3})$-approximation was given due to Kortsarts, Kortsarz and Nutov \cite{aprxdirectedmulticut3}. Finally \cite{agarwal2007improved} also gives an improved $\tilde{O}(n^{11/23})$-approximation algorithm for this problem. Again for graphs of treewidth $t$ a $t \log^{O(1)} n $ approximation algorithm was also shown in \cite{memoli2016quasimetric}.

On the hardness side \cite{chuzhoy2006hardness} demonstrated an $\Omega(\frac{\log{n}}{\log{\log{n}}})$-hardness for the Directed Non-Bipartite \spcut problem and the Directed \mcut problem under the assumption that \nptime $\not\subseteq$ \dtime $(n^{\log{n}^{O(1)}})$. This was further improved by them in a subsequent work \cite{chuzhoy2009polynomial} to obtain an $2^{\Omega(\log^{1 - \varepsilon}{n})}$-hardness result for both problems for any constant $\varepsilon > 0$ assuming that \nptime $\subseteq$ \zpptime.

Our main results for these problems are the following theorems.

\begin{theorem}\label{thm:spcutresults}
There exists a polynomial time $O(1)$-approximation algorithm for the Uniform Directed \spcut problem on series parallel graphs and graphs of bounded pathwidth. 
\end{theorem}

\begin{theorem}\label{thm:multicutresults}
There exists a polynomial time $O(1)$-approximation algorithm for the Directed \mcut problem on series parallel graphs and graphs of bounded pathwidth. 
\end{theorem}

We remark that in the above results the running time in the case of graphs of pathwidth $k$ is $n^{O(1)}$.
That is, the running time does not depend on $k$.
Typically, algorithms for graphs of pathwidth $k$ have running time of the form either $f(k)n^{O(1)}$, or $n^{f(k)}$, for some function $f$, due to the use of dynamic programming.
Our algorithms are based on LP relaxations, and thus avoid this overhead.

\subsection{Quasimetric spaces and embeddings}

\paragraph{Random quasipartitions.}
A Quasimetric space is a pair $(X,d)$ where $X$ is a set of points and $d:X\times X\to \mathbb{R}_{+}\cup\{+\infty\}$, that satisfies the following two conditions:

\begin{description}
\item{(1)}
For all $x,y\in X$, $d(x,y)=0$ iff $x=y$.
\item{(2)}
For all $x,y,z\in X$, $d(x,y)\leq d(x,z)+d(z,y)$.
\end{description}

The notion of random quasipartitions was introduced in \cite{memoli2016quasimetric}. A quasipartition is a transitive reflexive relation. Let $M=(X,d)$ be a quasimetric space. For a fixed $r \geq 0$, we say that a quasipartition $Q$ of $M$ is $r$-bounded if for every $x,y \in X$ with $(x,y) \in Q$, we have $d(x,y) \leq r$. Let ${\cal D}$ be a distribution over $r$-bounded quasipartitions of $M$.
We say that ${\cal D}$ is $r$-bounded.
Let $\beta > 0$.
We say that ${\cal D}$ is \emph{$\beta$-Lipschitz} if for any $x,y\in X$, we have that 
\[
\Pr_{P\sim {\cal D}}[(x,y)\notin P] \leq \beta \frac{d(x,y)}{r}.
\]

Given a distribution $\mathcal{D}$ over quasipartitions we sometimes use the term random quasipartition (with distribution $\mathcal{D}$) to refer to any quasipartition $P$ sampled from $\mathcal{D}$.
We consider the quasimetric space obtained from the shortest path distance of a directed graph. M\'emoli, Sidiropoulos and Sridhar in \cite{memoli2016quasimetric} find an $O(1)$-Lipschitz distribution over $r$-bounded quasipartitions of tree quasimetric spaces. They also prove the existence of a $O(t\log n)$-Lipschitz distribution over $r$-bounded quasipartitions for any quasimetric that is obtained from a directed graph of treewidth $t$.

Our main results for finding Lipschitz quasipartitions are the following theorems.

\begin{theorem}\label{thm:tw2}
Let $G$ be a directed graph of treewidth 2.
Let $M = (V(G),d_G)$ denote the shortest-path quasimetric space of $G$.
Then for all $r > 0$, there exists some $O(1)$-Lipschitz distribution over $r$-bounded quasipartitions of $M$.
\end{theorem}

\begin{theorem}\label{thm:boundedpathwidth}
Let $G$ be a directed graph of of pathwidth $k$.
Let $M = (V(G),d_G)$ denote the shortest-path quasimetric space of $G$. Then for all $\Delta > 0$, there exists some $2^{O(k^2)}$-Lipschitz distribution over $\Delta$-bounded quasipartitions of $M$.
\end{theorem}


\paragraph{Random embeddings.}
Before stating our embedding results, we first need to introduce some notations and definitions. Let $M=(X,d)$ and $M'=(X',d')$ be quasimetric spaces.
A mapping $f:X\to X'$ is called an \emph{embedding of distortion $c\geq 1$} if there exists some $\alpha>0$, such that for all $x,y\in X$, we have 
$d(x,y)\leq \alpha \cdot d'(f(x),f(y)) \leq c\cdot d(x,y)$. 
We say that $f$ is \emph{isometric} when $c=1$.
Let ${\cal D}$ be a distribution over pairs $(M',f)$, where $f:X\to X'$.
We say that ${\cal D}$ is a random embedding of distortion $c\geq 1$ if for all $x,y\in X$, the following conditions are satisfied:
\begin{description}
\item{(1)}
 $\Pr_{(M',f)\sim {\cal D}}[d'(f(x),f(y)) \geq d(x,y)] = 1$.
\item{(2)}
 $\mathbf{E}_{(M',f)\sim {\cal D}}[d'(f(x),f(y))]\leq c\cdot d(x,y)$.
\end{description}

\paragraph{Directed $\ell_1$ (Charikar \etal~ \cite{charikar2006directed})}
 The directed $\ell_1$ distance between two points $x$ and $y$ is given by $d_{\ell_1}(x,y) = \sum\limits_i |x_i -y_i| + \sum\limits_i |x_i| - \sum\limits_i |y_i| $. 

The following theorems are our main results for random embeddings.

\begin{theorem}\label{thm:cycle_ell1}
Let $G$ be a directed cycle and let $M=(V(G),d_G)$ be the shortest-path quasimetric space of $G$.
Then $M$ admits a constant-distortion embedding into directed $\ell_1$.
Moreover the embedding is computable in polynomial time.
\end{theorem}

\begin{theorem}\label{thm:treel1}
Let $G=(V,E)$ be a directed tree, and let $X=(V,d)$ be the quasimetric induced by $G$. Then $X$ embeds into directed $\ell_1$.
\end{theorem}

\paragraph{Limitations.}
We further discuss some limitations of methods that were developed for undirected graphs. Klein, Plotkin, and Rao in \cite{klein1993excluded} introduced the notion of random partitions for undirected graphs. In Section \ref{sec:limits}, we show that this method can not be used or generalized for the case of directed graphs. Furthermore, we complete our paper with a lower bound result that is stated in the following theorem.

\begin{theorem}\label{thm:lowerbound}
There exists a directed cycle $G = (V,E)$ such that any non-contracting random embedding of $G$ into directed trees has distortion $\Omega(n)$. 
\end{theorem}


\subsection{Organization}
In Sections \ref{sec:treewidth} and \ref{sec:pathwidth}, we provide efficient algorithms for computing random quasipartitions for directed graphs of treewidth 2 and bounded pathwidth graphs respectively. In Section \ref{sec:cyclel1}, we describe an algorithm for computing an $O(1)$-distortion embedding of the directed cycles into directed $\ell_1$. In Section \ref{sec:treel1}, we provide an algorithm for embedding directed trees into directed $\ell_1$ with distortion one. In Section \ref{sec:application} we discuss the applications to directed cut problems.
In Section \ref{sec:limits} we discuss the limitations of random partitions for the directed case, and finally in Section \ref{sec:lowerbound} we provide a lower for non-contracting embeddings of directed cycle into directed trees.

\section{Notation and preliminaries}
We now introduce some notation that will be used throughout the paper.

\paragraph{Directed graphs and treewidth.}
From any undirected graph $G^\UN$ we can obtain a directed graph $G$ as follows.
We set $V(G) = V(G^\UN)$ and $E(G) = \{(u,v),(v,u) : \{u,v\} \in E(G^\UN) \}$; i.e.~for every undirected edge $\{u,v\}$ in $E(G^\UN)$ we add directed edges $(u,v)$ and $(v,u)$ to $E(G)$.
We refer to $G^\UN$ as the \emph{underlying undirected graph} of $G$.
We say that $G$ is a directed graph of treewidth $k$ if its underlying undirected graph has treewidth $k$, for some $k\geq 1$.
Similarly, we say that $G$ is a directed tree (resp.~directed cycle) if its underlying undirected graph is a tree (resp.~cycle).

\paragraph{Directed cut metrics and 0-1 quasimetrics (Charikar \etal~ \cite{charikar2006directed}) }

Given a set $X$ and a subset $S \subset X$, the corresponding \emph{directed cut metric} distance for any pair of elements $u,v \in X$ is given by,
\begin{align*}
d_S(x,y) &= \left\{\begin{array}{ll}
1 & \text{ if } x \in S, y \not\in S ;\\
0 & \text{ otherwise };
\end{array}\right.
\end{align*}
 
A \emph{0-1 Quasimetric space} is a pair $(X,d)$ where for all $u,v \in X$ we have that $d(u,v) = 0$ or $d(u,v) =1$ and for all $u,v,w \in X$ we have that $d(u,w) \leq d(u,v) + d(v,w)$.

\section{Lipschitz quasipartitions of treewidth-2 directed graphs}\label{sec:treewidth}

In this Section we provide a proof for Theorem \ref{thm:tw2}. Note that since all series-parallel graphs have treewidth at most $2$ this result automatically holds for any series-parallel graph. We present an efficient algorithm for computing a random quasipartition of a directed graph of treewidth 2.
We begin by describing some special type of graphs of treewidth 2, which we refer to as \emph{trees of hexagons}.
We show that any graph of treewidth 2 admits an isometric embedding into some tree of hexagons.
We then further show how to preprocess a tree of hexagons such that it can be inductively constructed via a sequence of either \emph{slack} or \emph{tight} paths similar to \cite{gupta2004cuts}.
Finally, we present the algorithm for computing the random quasipartition, and we analyze the correctness of the algorithm.

\subsection{Trees of hexagons}
Let $G$ be a directed graph of treewidth $2$. We can construct $G^\UN$ as follows.
Start with a single edge and sequentially perform the following operation. Pick an arbitrary existing edge $e = \{u,v\}$. Add a new vertex $x$ and edges $\{x,u\}$ and $\{x,v\}$.
Let $\Gamma = (u,x,v)$ be the added path.
We say that $e$ is the \emph{parent} of $\Gamma$. Finally, remove an arbitrary subset of edges. Now we may assume w.l.o.g.~that no edges are removed in the last step while constructing $G^\UN$. Suppose this is not the case then we can replace the removed edges with edges that have weight equal to the shortest path distance between the two end points. This will ensure that the induced shortest path quasimetric of $G$ remains the same. The parent relation induces a rooted tree decomposition $({\cal T}, {\cal X})$ of $G^\UN$ of width $2$, where each bubble induces a triangle in $G^\UN$. 

For any path $\Gamma = (u,x,v)$ of $G^\UN$ whose parent edge is $e = \{u,v\}$ we say that the directed edge $(u,v)$ is the parent of the directed path $(u,x,v)$ in $G$ and the directed edge $(v,u)$ is the parent of the directed path $(v,x,u)$ in $G$.

We construct a new graph $G'$ as follows. We start with $G$ and modify it in the following fashion. For all $B \in {\cal X}$ we consider the sub-graph $G[B]$ and proceed as follows. Let $u,v,w$ be the vertices of $B$. We duplicate each vertex of $B$ and add edges of weight $0$ in both directions between the two copies of a vertex (See Figure \ref{fig:Hexagons}). By doing so, every directed triangle in $G[B]$ corresponds to a directed hexagon in $G'$, where $u',u'',v',v'',w',w''$ are the vertices of the hexagon. The edges $e_1=(u,v)$, $e_2=(v,w)$, $e_3=(w,u)$ correspond to $e'_1=(u'',v')$, $e'_2=(v'',w')$, and $e'_3=(w'',u')$ respectively in $G'$ with the same weight. Similarly, the edges $e_4=(v,u)$, $e_5=(w,v)$, and $e_6=(u,w)$ correspond to $e'_4=(v',u'')$, $e'_5=(w',v'')$, and $e'_6=(u',w'')$ respectively in $G'$ with the same weight (See Figure \ref{fig:Hexagons}). Therefore, we get a new graph $G'$.

For any triangle $\{u,w,v\}$ in $G$ where the edge $e = \{u,v\}$ is the parent of the path $\Gamma = \{u,w,v\}$, let the corresponding hexagon in $G'$ be $\{u',u'',w',w'',v',v''\}$. We call the directed edge $e' = (u'',v')$ the \emph{parent} edge of the directed path $\Gamma' = \{u'',u',w'',w',v'',v'\}$ and of every edge in it. Similarly we call $e'' = (v',u'')$ the parent edge of the directed path $\Gamma'' = \{v',v'',w',w'',u',u''\}$ and every edge in it.
This parent relation induces a rooted tree decomposition of $G'^\UN$ , $({\cal T'}, {\cal X'})$ of width $5$, where each bubble $B \in {\cal X'}$ induces a hexagon in $G'^\UN$.  We say that $G'$ is a \emph{tree of hexagons}.

Let $e_1 = (u_1,v_1)$ be the parent edge of a path $\Gamma_1$. For any edge $e_0 = (u_0,v_0)$ in $\Gamma_1$ we define the \emph{tail} of $e_0$, $\tail(u_0,v_0)$, to be the subpath of $\Gamma_1$ from $u_1$ to $v_0$. 

The above discussion immediately implies the following.

\begin{lemma}[Embedding into a tree of hexagons]\label{lem:embedding_hexagon}
There exists a polynomial-time algorithm which given a directed graph $G$ of treewidth 2, computes an isometric embedding of $G$ into some tree of hexagons $G'$. 
\end{lemma}

\begin{figure}
\begin{center}
\scalebox{0.30}{\includegraphics{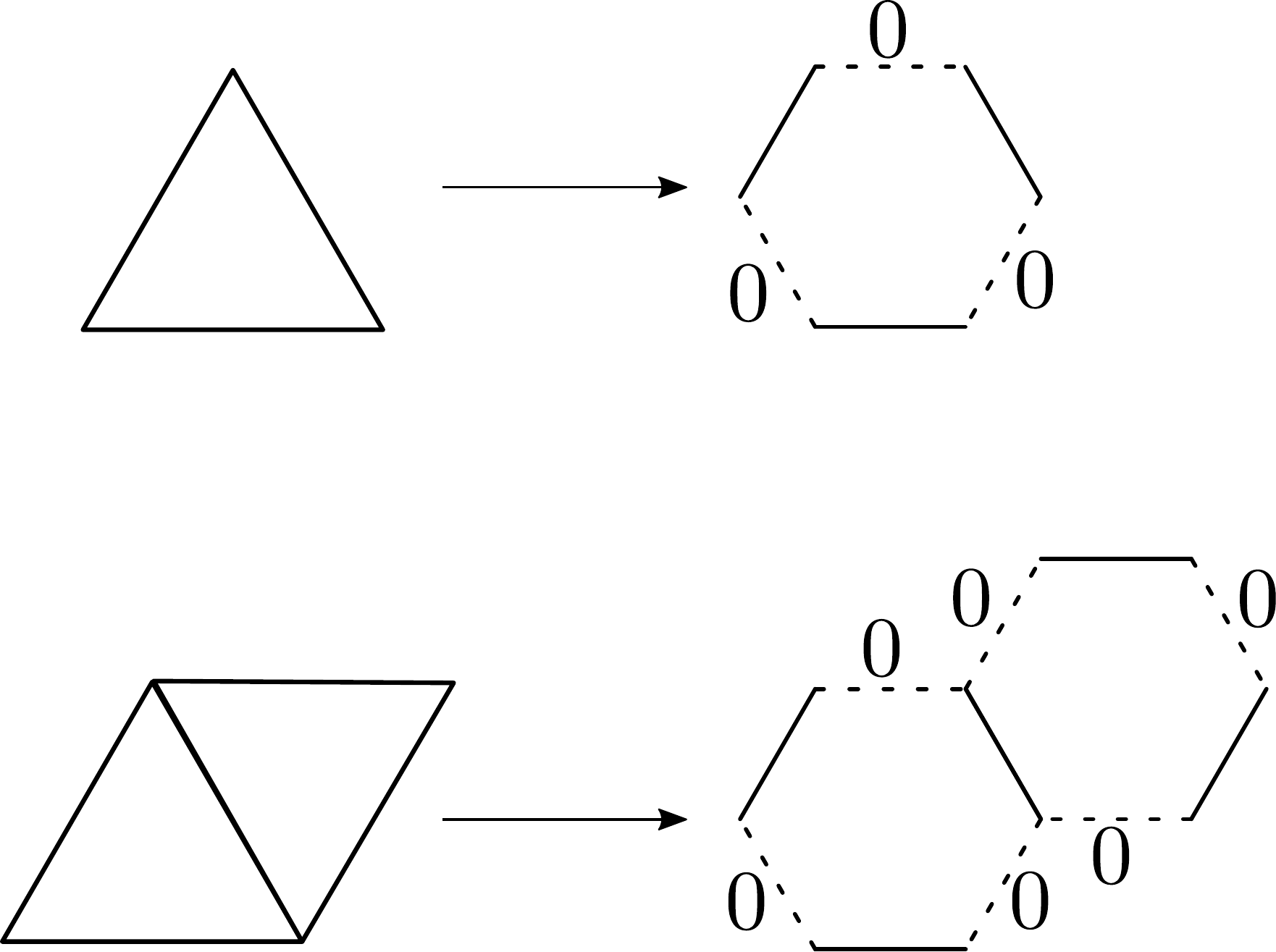}}
\caption{Hexagons}
\label{fig:Hexagons}
\end{center}
\end{figure}

\subsection{Slack and tight paths}
Let $G$ be the input graph, and let $w$ be a non-negative weight function on the edges of $G$. By Lemma \ref{lem:embedding_hexagon} we may assume w.l.o.g.~that $G$ is a tree of hexagons, and we have a tree decomposition of $G'^\UN$ , $({\cal T}, {\cal X})$ rooted at some $B^* \in {\cal X}$. 
For any two vertices $x,y \in V(G)$, we pick a unique shortest path from $x$ to $y$ denoted by $p_{xy}$ to use in our algorithm. We always pick $p_{xy}$ to be a shortest path with the fewest number of edges. If there are multiple such paths we pick one maintaining the condition that the intersection of any two shortest paths is a (possiply empty) path. For a path $P$, let $\len(P)$ denote the length of $P$. Let $e \in E(G)$. For any child path $\Gamma$ of $e$, we say that $\Gamma$ is \emph{slack} if $\len(\Gamma) \geq 2w(e)$, and we say that $\Gamma$ is \emph{tight} if $\len(\Gamma) = w(e)$.

Let $x \in B^*$ be an arbitrary vertex. Let $l \colon {\cal X} \to \mathbb{Z}_{\geq 0}$ be a level function where $l(B^*) = 0$, and for any other $B \in {\cal X}$, $l(B)$ denotes the length of the shortest path from $B^*$ to $B$ in ${\cal T}$. 
Let $B_1, B_2, \ldots, B_k$ be the leaf vertices of ${\cal T}$. For every $i \in {1,2,\ldots,k}$, let $l_i \in B_i$ be an arbitrary vertex in $G$. For every $i \in {1,2,\ldots,k}$, let $P_i$ be the unique shortest paths in ${\cal T}$ from $B^*$ to $B_i$.
 Let $p_i = p_{xl_i}$ and let $q_i = p_{l_{i}x}$. For every $p_i$, we define the \emph{complementary} path $\hat{p_i}$ as follows. Let $P_i$ be the path of hexagons corresponding to $p_i$. Let $P'_i$ be the subgraph of $P_i$ obtained by deleting all the edges of $p_i$, i.e. $P'_i = P_i \setminus p_i$. We set $\hat{p_i}$ to be the unique path from $x$ to $l_i$ in $P'_i$ (See Figure \ref{fig:complementary}). For every $p_i$, we define the \emph{flattened complementary} graph $\overline{p_i}$ as follows. Start with $\overline{p_i}=\hat{p_i}$, and for every tight path $\Gamma$ with a parent edge $e \in \overline{p_i}$ add $\Gamma$ to $\overline{p_i}$ and repeat until we don't add any new paths. We can similarly define the \emph{complementary} path $\hat{q_i}$ and the \emph{flattened complementary} graph $\overline{q_i}$ for every $q_i$.

\begin{figure}
\begin{center}
\scalebox{0.2}{\includegraphics{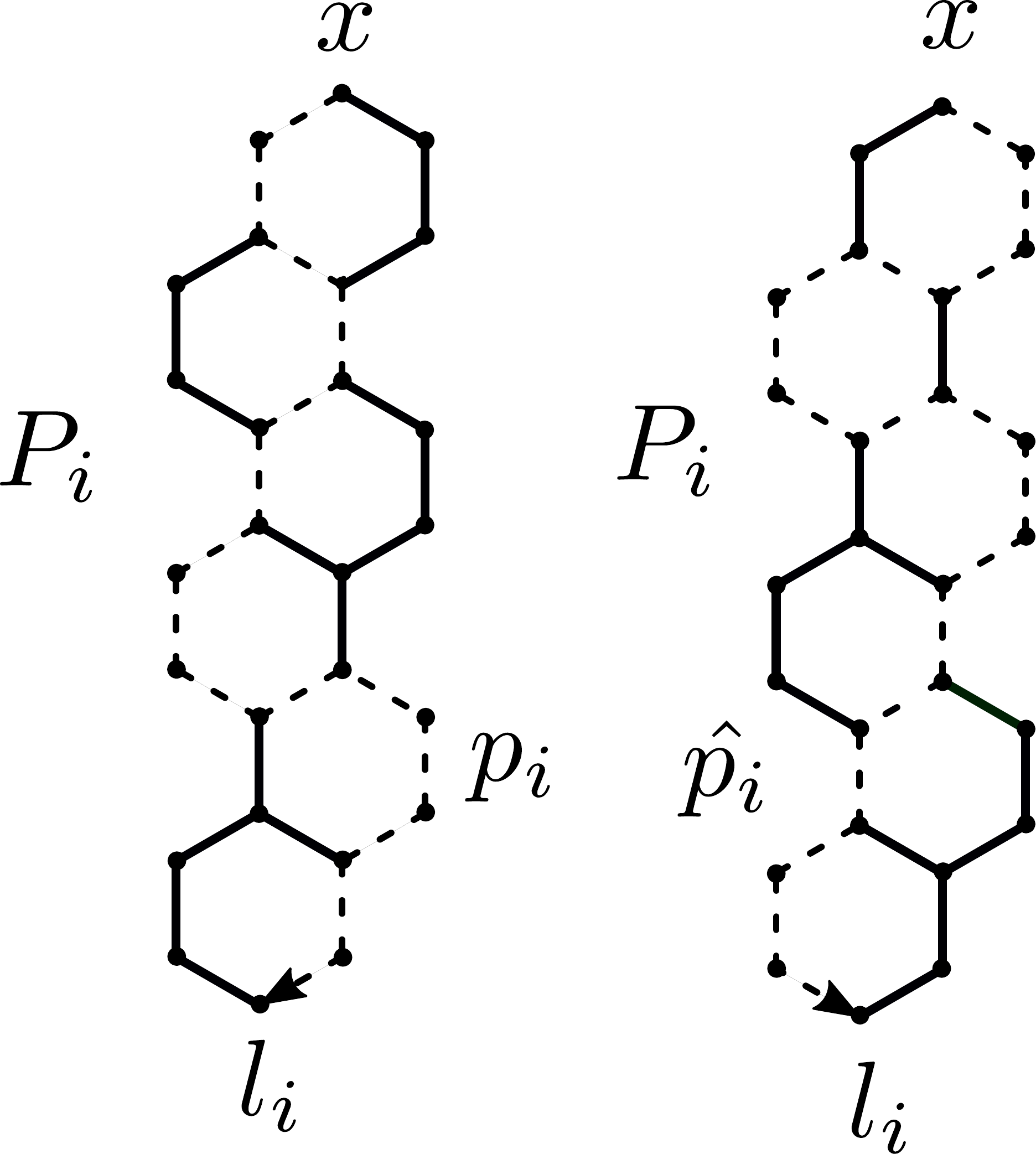}}
\caption{Complementary path}
\label{fig:complementary}
\end{center}
\end{figure}

Let $P$ be a path in $G$.
We say that $P$ is \emph{down-monotone} if when traversing $P$ we visit the bubbles of ${\cal T}$ in non-decreasing distance from the root of ${\cal T}$.
Similarly, we say that $P$ is \emph{up-monotone} if when traversing $P$ we visit the bubbles of ${\cal T}$ in non-increasing distance from the root of ${\cal T}$.


We say that some tree of hexagons $G$ is \emph{canonical} if for all $e\in E(G)$, every child of $e$ is either tight or slack.
We first show that any directed graph of treewidth 2 admits a constant-distortion embedding into a canonical directed tree of hexagons.
This allows us to focus on canonical graphs.

\begin{lemma}[Embedding into a canonical graph]
Given a directed tree of hexagons $G$, we can compute in polynomial time some embedding of $G$ into some canonical tree of hexagons $G'$, with distortion at most $2$.
\end{lemma}

\begin{proof}
The algorithm proceeds by inductively modifying the graph $G$.
We intially mark all edges as \emph{unresolved}.
We mark all edges with no parent as \emph{resolved}.
While there are unresolved edges, we pick some unresolved edge $e$, whose parent $e'$ is resolved, and let $\Gamma$ be the child path of $e'$ that contains $e$.
If $\Gamma$ is neither slack nor tight then for all $e\in E(\Gamma)$, we set $w(e)=w(e) \cdot w(e')/\len(\Gamma)$.
We all mark edges in $\Gamma$ as resolved.
We set $G'$ be the graph obtained at the end of this inductive process.
It is immediate the $G'$ is canonical.
At each iteration the number of unresolved edges decreases by at least one, so the algorithm terminates in polynomial time.
By the definition of tight and slack paths, it follows that the length of each edge changes by at most a factor of 2.
Thus the distortion of the induced embedding is at most 2, which concludes proof.
\end{proof}

\subsection{Computing a random quasipartition}

The algorithm for computing a random quasipartition is as follows.
The input consists of some directed tree of hexagons $G$, a non-negative weight function $w$ on the edges of $G$, and some $r>0$.
The output is a random $r$-bounded quasipartition of the shortest-path quasimetric space of $G$.


\begin{description}
\item{\textbf{Input:}}
A directed canonical tree of hexagons $G$, and a tree decomposition of $G^\UN$ , $({\cal T}, {\cal X})$ rooted at some $B^* \in {\cal X}$. and $r>0$.
\item{\textbf{Output:}}
Random quasipartition of the shortest-path quasimetric space of $G$, $(M,d_G)$.
\item{\textbf{Initialization.}}
Set $G^*=G$ and $Q=E(G)$.

\item{\textbf{Step 1.}} 
Let $x \in B^*$ be an arbitrary vertex. Pick $z \in [0,r]$ uniformly at random.

\item{\textbf{Step 2.}} For all $(u,v) \in E(G^*)$ remove $(u,v)$ from $Q$ if $d_G(x,v)> i\cdot r + z$ and $d_G(x,u) \leq i\cdot r + z$ for some integer $i \geq 0$.
\item{\textbf{Step 3.}} For all $(u,v) \in E(G^*)$ remove $(u,v)$ from $Q$ if $d_{\overline{{p_j}}}(x,v)> i\cdot r + z$ and $d_{\overline{{p_j}}}(x,u) \leq i\cdot r + z$ for some integer $i \geq 0$ and some integer $j \in \{0,\ldots,k \} $.

\item{\textbf{Step 4.}} 
For all $(u,v) \in E(G^*)$ that are removed from $Q$ in Step 3 do the following:

\begin{description}
\item{\textbf{Step 4.1.}} For each uncut child path $\Gamma = (u,w,q,o,p,v)$ of $(u,v)$ remove one of the edges $(u,w)$, $(w,q)$, $(q,o)$, $(o,p)$ or $(p,v)$ from $Q$, chosen randomly with probability $\frac{d(u,w)}{\len(\Gamma)}$, $\frac{d(w,q)}{\len(\Gamma)}$, $\frac{d(q,o)}{\len(\Gamma)}$, $\frac{d(o,p)}{\len(\Gamma)}$ and $\frac{d(p,v)}{\len(\Gamma)}$ respectively.

\item{\textbf{Step 4.2.}} Recursively perform Step 4.1 on the removed edge.
\end{description}

\item{\textbf{Step 5.}} For all $(u,v) \in E(G^*)$ remove $(u,v)$ from $Q$ if $d(v,x)\leq i\cdot r + z$ and $d(u,x) > i\cdot r + z$ for some integer $i \geq 0$.

\item{\textbf{Step 6.}} For all $(u,v) \in E(G^*)$ remove $(u,v)$ from $Q$ if $d_{\overline{{q_j}}}(v,x)\leq i\cdot r + z$ and $d_{\overline{{q_j}}}(u,x) > i\cdot r + z$ for some integer $i \geq 0$ and some $j\in \{0,\ldots,k\}$.

\item{\textbf{Step 7.}} 
For all $(u,v) \in E(G^*)$ that are removed from $Q$ in step 6 do the following:
\begin{description}
\item{\textbf{Step 7.1.}} For each uncut child path $\Gamma = (u,w,q,o,p,v)$ of $(u,v)$ remove one of the edges $(u,w)$, $(w,q)$, $(q,o)$, $(o,p)$ or $(p,v)$ from $Q$, chosen randomly with probability $\frac{d(u,w)}{\len(\Gamma)}$, $\frac{d(w,q)}{\len(\Gamma)}$, $\frac{d(q,o)}{\len(\Gamma)}$, $\frac{d(o,p)}{\len(\Gamma)}$ and $\frac{d(p,v)}{\len(\Gamma)}$ respectively.

\item{\textbf{Step 7.2.}} Recursively perform Step 7.1 on the removed edge.
\end{description}

\item{\textbf{Step 8.}} For any $(u,v) \in Q$, if $d(u,v) > \frac{r}{10}$, remove $(u,v)$ from $Q$.

\item{\textbf{Step 9.}} Enforce transitivity on $Q$; that is, for all $u,v,w \in V(G)$ if $(u,v) \in Q$ and $(v,w) \in Q$ then add $(u,w)$ to $Q$.
\end{description}

This concludes the description of the algorithm for computing a random quasipartition.

\paragraph{Analysis.}

We now analyze the performance of the above algorithm.
We begin by showing that the probability that an edge is removed from the quasipartition is small.
This statement is shown by considering separately all possible steps of the algorithm where an edge can be removed. 

\begin{lemma}\label{lem:step2}
For all $(u,v) \in E(G)$, we have $Pr[(u,v) \text{ is removed from $Q$ in Step 2}] \leq  \frac{d(u,v)}{r}$.
\end{lemma}

\begin{proof}
The edge $(u,v)$ is removed from $Q$ in Step 2 when $d(x,u) \leq  z + ir $ and $d(x,v) > z + ir$ for some integer $i\geq 0$.
By the triangle inequality this implies that $Pr[(u,v) \text{ is removed from $Q$ in Step 2}] \leq \frac{d(x,v) - d(x,u)}{r} \leq \frac{d(u,v)}{r}$, as required.
\end{proof}

\begin{lemma}\label{lem:step5}
For all $(u,v) \in E(G)$, we have $Pr[(u,v) \text{ is removed from $Q$ in step 5}] \leq  \frac{d(u,v)}{r}$.
\end{lemma}

\begin{proof}
The proof of this case is similar to the proof of Lemma \ref{lem:step2}.
\end{proof}

\begin{lemma} \label{lem:spines}
Let $(u,v) \in E(G)$. Suppose that $(u,v) \in \hat{p_i}$ and $(u,v) \in  \hat{p_j}$ for some $i,j \in \{0,\ldots,k \}$, then $d_{\overline{p_i}}(x,u) = d_{\overline{p_j}}(x,u)$ and $d_{\overline{p_i}}(x,v) = d_{\overline{p_j}}(x,v)$. 
\end{lemma}

\begin{proof}
Let $\tau$ be the hexagon containing $u$ and $v$. Let $(y,o)$ be the parent edge of $(u,v)$. Note that it is possible that $y = u$ or $o = u$. Since any complement subpath from $x$ to $(u,v)$ must end with the unique subpath $\tail(u,v)$, it follows that $\tail(u,v) \in \hat{p_i}$ and $\tail(u,v) \in \hat{p_j}$. The fact that the complement contains $y$ also means that $p_{xo}$ is contained in $p_i$ and $p_j$. Therefore we have that $\hat{p_i}$ and $\hat{p_j}$ share the same sub-path from $x$ to $y$. Combined with the fact that $\tail(u,v) \in \hat{p_i}$ and $\tail(u,v) \in \hat{p_j}$ this implies that $\hat{p_i}$ and $\hat{p_j}$ share the same sub-path from $x$ to $v$ ending with the edge $(u,v)$.

Since $\hat{p_i}$ and $\hat{p_j}$ share the same sub-path from $x$ to $v$ ending with the edge $(u,v)$ we have that $d_{\overline{p_i}}(x,u) = d_{\overline{p_j}}(x,u)$ and $d_{\overline{p_i}}(x,v) = d_{\overline{p_j}}(x,v)$.
\end{proof}

\begin{lemma}\label{lem:tight}
Let $i \in \{0,\ldots,k\}$. Suppose that $p_i$ traverses the parent edge $(u,v)$ of a tight path $\Gamma$. Then $p_i$ does not visit any vertex in $\Gamma$ other than $u$ and $v$.
\end{lemma}

\begin{proof}
Let $\Gamma = (u,o,w,f,y,v)$. Since $(u,v)$ is the parent edge of a tight path $(u,o,w,f,y,v)$ we have that $p_{uy} = (u,o,w,f,y)$. Suppose $p_i$ visits some vertex in $\Gamma$ other than $u$ and $v$ we have that $p_i$ intersects the shortest path $p_{uy}$ more than once which is a contradiction.
\end{proof}

\begin{lemma}\label{lem:spine and closure}
Let $(u,v) \in E(G)$. Suppose that $(u,v) \in \overline{p_i}$, $(u,v) \not\in \hat{p_i}$  and $(u,v) \in  \hat{p_j}$ for some $i,j \in \{0,\ldots,k \}$, then $d_{\overline{p_i}}(x,u) = d_{\overline{p_j}}(x,u)$ and $d_{\overline{p_i}}(x,v) = d_{\overline{p_j}}(x,v)$. 
\end{lemma}

\begin{proof}
Let $(t,z) \in \hat{p_i}$ be the unique ancestor edge of $(u,v)$ that is contained in $\hat{p_i}$.
We have that $(t,z)$ is the ancestor edge of a tight path that contains $(u,v)$. This implies that $p_{tv} \in \overline{p_i}$. Now let $(b,c)$ be the parent edge of $(t,z)$. This implies that $p_{xc} \subset p_i$.

Let us suppose that the parent edge of $(u,v)$ is $(y,o)$. This implies that $p_j$ contains $o$. Since $(y,o)$ is the parent edge of a tight path we also have that $p_{yv} = \tail(u,v)$. This implies that $p_{tv} = p_{ty} \cup  \tail(u,v)$.
From Lemma \ref{lem:tight} we have that $p_j$ does not contain $(y,o)$. Recursively applying Lemma \ref{lem:tight} we have that $p_j$ does not intersect with $(\tail(t,z)\setminus (t,z)) \cup p_{tv}$. Therefore we have that $p_{xc} \subset p_j$ and that $(\tail(t,z)\setminus (t,z)) \cup p_{tv} \subset \hat{p_j}$. Since $p_{xc}$ is common to both $p_i$ and $p_j$, we have that $\overline{p_i}$ and $\overline{p_j}$ share the same subpath from $x$ to $b$. Moreover they both also contain $(\tail(t,z)\setminus (t,z)) \cup p_{tv}$ which concludes the proof.
\end{proof}

\begin{lemma}\label{lem:closure}
Let $(u,v) \in E(G)$. Suppose that $(u,v) \in \overline{p_i}$ and $(u,v) \in  \overline{p_j}$ for some $i,j \in \{0,\ldots,k \}$, then $d_{\overline{p_i}}(x,u) = d_{\overline{p_j}}(x,u)$ and $d_{\overline{p_i}}(x,v) = d_{\overline{p_j}}(x,v)$. 
\end{lemma}

\begin{proof}
Let $(t,z) \in \hat{p_i}$ be the unique ancestor edge of $(u,v)$ that is contained in $\hat{p_i}$. Let $(b,c) \in \hat{p_j}$ be the unique ancestor edge of $(u,v)$ that is contained in $\hat{p_j}$. 
W.l.o.g.~let $(t,z)$ be the ancestor edge of a tight path that contains $(b,c)$. From Lemma \ref{lem:spine and closure} we have that $d_{\overline{p_i}}(x,b) = d_{\overline{p_j}}(x,b)$. Since $(b,c)$ is an ancestor edge of a tight path that contains $(u,v)$ it follows that $\overline{p_i}$ and $\overline{p_j}$ both contain $p_{bv}$. This implies that $d_{\overline{p_i}}(x,u) = d_{\overline{p_j}}(x,u)$ and $d_{\overline{p_i}}(x,v) = d_{\overline{p_j}}(x,v)$ concluding the proof. 
\end{proof}

\begin{lemma}\label{lem:step3}
For all $(u,v) \in E(G)$, we have $Pr[(u,v) \text{ is removed from $Q$ in Step 3}] \leq \frac{d(u,v)}{r}$.
\end{lemma}

\begin{proof}
From Lemmas \ref{lem:spines}, \ref{lem:spine and closure} and \ref{lem:closure} we have that $(u,v)$ is only removed when, for some integer $i\geq 0$, and any $j \in \{0,\ldots,k \}$ such that $(u,v)$ is in $\overline{p_j}$, we have that $d_{\overline{p_j}}(x,u) \leq  z + ir $ and $d_{\overline{p_j}}(x,v) > z + ir$ .
By the triangle inequality, this implies that $Pr[(u,v) \text{ is removed from $Q$ in Step 3}] \leq \frac{d(x,v) - d(x,u)}{r} \leq \frac{d(u,v)}{r}$, as required. 
\end{proof}

\begin{lemma}\label{lem:step6}
For all $(u,v) \in E(G)$, we have $Pr[(u,v) \text{ is removed from $Q$ in Step 6}] \leq \frac{d(u,v)}{r}$.
\end{lemma}

\begin{proof}
The proof of this case is similar to the proof of Lemma \ref{lem:step3}.
\end{proof}

\begin{lemma}\label{lem:step4}
For all $(u,v) \in E(G)$, we have $Pr[(u,v) \text{ is removed from $Q$ in Step 4}] \leq 2 \frac{d(u,v)}{r}$.
\end{lemma}

\begin{proof}
We prove this by induction. For the base case suppose that $(u,v)$ is an edge in $B^*$ then it has no parent edge and therefore the assertion is immediate.
Otherwise let $e$ be the parent edge of a child path $\Gamma$ containing $(u,v)$ and assume, by the inductive hypothesis, that $Pr[(e) \text{ is removed from $Q$ in Step 4}] \leq 2 \frac{d(e)}{r}$. There are two cases:
\begin{description}
\item{\emph{Case 1:}} Suppose that $\Gamma$ is a tight child path of $e$. Then we have 
\begin{align*}
Pr[(u,v) \text{ is removed from $Q$ in Step 4}] &= Pr[(e) \text{ is removed from $Q$ in Step 4}] \cdot \frac{d(u,v)}{d(e)} \\ 
& \leq 2 \frac{d(u,v)}{r}
\end{align*} 

\item{\emph{Case 2:}} Suppose that $\Gamma$ is a slack child path of $e$. Then we have
\begin{align*}
Pr[(u,v) \text{ is removed from $Q$ in Step 4}] &= Pr[(e) \text{ is removed from $Q$ in Step 4}] \cdot \frac{d(u,v)}{\len_G(\Gamma)} \\
& ~~~ + Pr[(e) \text{ is removed from $Q$ in Step 3}] \cdot \frac{d(u,v)}{\len_G(\Gamma)} \\ & \leq 2\frac{d(e)}{r} \cdot \frac{d(u,v)}{2d(e)} + \frac{d(e)}{r}\cdot \frac{d(u,v)}{2d(e)} \\
& \leq 2 \frac{d(u,v)}{r}
\end{align*}
\end{description} 
Thus in either case the assertion is satisfied, concluding the proof.
\end{proof}

\begin{lemma}\label{lem:step7}
For all $(u,v) \in E(G)$, we have $Pr[(u,v) \text{ is removed from $Q$ in Step 7}] \leq O(1) \frac{d(u,v)}{r}$.
\end{lemma}

\begin{proof}
The proof for this is similar to the proof of Lemma \ref{lem:step4}.
\end{proof}

\begin{lemma}\label{lem:step8}
For all $(u,v) \in E(G)$, we have $Pr[(u,v) \text{ is removed from $Q$ in Step 8}] \leq 10 \frac{d(u,v)}{r}$.
\end{lemma}

\begin{proof}
Since only edges of length greater than $\frac{r}{10}$ are removed in Step 8 we have that
\[
Pr[(u,v) \text{ is removed from $Q$ in step 8}] \leq 1 < 10 \frac{d(u,v)}{r}.
\]
\end{proof}

\begin{lemma}\label{lem:edge_probtw2}
For all $(u,v) \in E(G)$, we have $Pr[(u,v) \text{ is removed from $Q$}] \leq 18 \frac{d(u,v)}{r}$.
\end{lemma}

\begin{proof}
The assertion follows by combining Lemmas \ref{lem:step2}, \ref{lem:step3}, \ref{lem:step4}, \ref{lem:step5}, \ref{lem:step6}, \ref{lem:step7} and \ref{lem:step8} using the union bound. 
\end{proof}

Finally we show that $Q$ is $6r$-bounded.

 \begin{lemma}\label{lem:monotone_path_down}
Let $e=(u,v) \in Q$ where $u \in B_u$ and $v \in B_v$. Suppose that $B_u$ is in the path from $B^*$ to $B_v$ in $\cal{T}$. If $(u,v) \in Q$ then there exists a monotone-down path $P = \{a_1=u,a_2,\ldots,a_t=v\}$ in $G$, where for all $i \in \{1,2,\ldots,t-1\}$ we have $(a_i,a_{i+1}) \in Q$.
\end{lemma}

\begin{proof}
Since $(u,v) \in Q$, we have that at the beginning of step $9$ there must have been a path $P_0 = \{a_1=u,a_2,\ldots,a_t=v\}$ such that for all $i \in \{1,2,\ldots,t-1\}$ we have $(a_i,a_{i+1}) \in Q$. If $P_0$ is monotone-down, we are done. Otherwise, we start with $P_0$ and modify it to obtain the desired $P$. Suppose that $P_0$ is not monotone-down. Let $X=(B_u,...,B_v)$ be the shortest path from $B_u$ to $B_v$ in $\cal{T}$. Let $A = \cup_{a\in B_i, B_i \in V(X)}a$. Since $P_0$ is not monotone-down, there exists $a_i \in V(P_0)$ such that $a_i \notin A$. Let $m \in \{1,2,\ldots,t\}$ be the smallest number such that $a_m$ has such property. Let $C'_{m,0}$ be the hexagon containing $a_{m-1}$ and $a_{m-2}$ (See Figure \ref{fig:monotone}). Let $a_s$ be the other neighbor of $a_{m-1}$ in $C'_{m,0}$, and let $e_0 = (a_{m-1},a_s)$. Let $C'_{m,1}$ be the next hexagon traversed by $P_0$ after $C'_{m,0}$, and let $e_1 \in C'_{m,1}$ be the other edge in $C'_{m,1}$ which is not traversed by $P_0$. We similarly define $C'_{m,2},\ldots,C'_{m,l}$, and $e_2,\ldots,e_{l-1}$. 

The main idea is that we are able to replace the subpath of $P_0$ from $a_{m-1}$ to $a_s$ with $e_0$. Suppose that we are not able to do such replacement, and thus $e_0$ is removed in the algorithm. Let $\Gamma = (a_{m-1},a_m,a_{m+1},a_{s-1},a_s)$. $\Gamma$ is a child of $e_0$. First suppose that $\Gamma$ is tight. In this case, since $e_0$ is removed, at least one of the edges of $\Gamma$ should be removed after step $3$, and thus $e_1$ should be removed by the algorithm. Second case is where $\Gamma$ is a slack. In this case, since $e_0$ is removed, at least one of the edges of $\Gamma$ should be removed after step $4$, and again we have that $e_1$ should be removed by the algorithm. Using the same argument inductively, we get that $e_{l-1}$ should be removed by the algorithm. But note that none of the edges of the child of $e_{l-1}$ in $C'_{m,l}$ are removed, which is a contradiction. Therefore, $e_0$ is not removed by the algorithm and we can replace the subpath of $P_0$ from $a_{m-1}$ to $a_s$ with $e_0$. Using the same argument, we can modify $P_0$ such that it only traverses vertices in $A$, as desired.
\end{proof}

\begin{figure}
\begin{center}
\scalebox{0.33}{\includegraphics{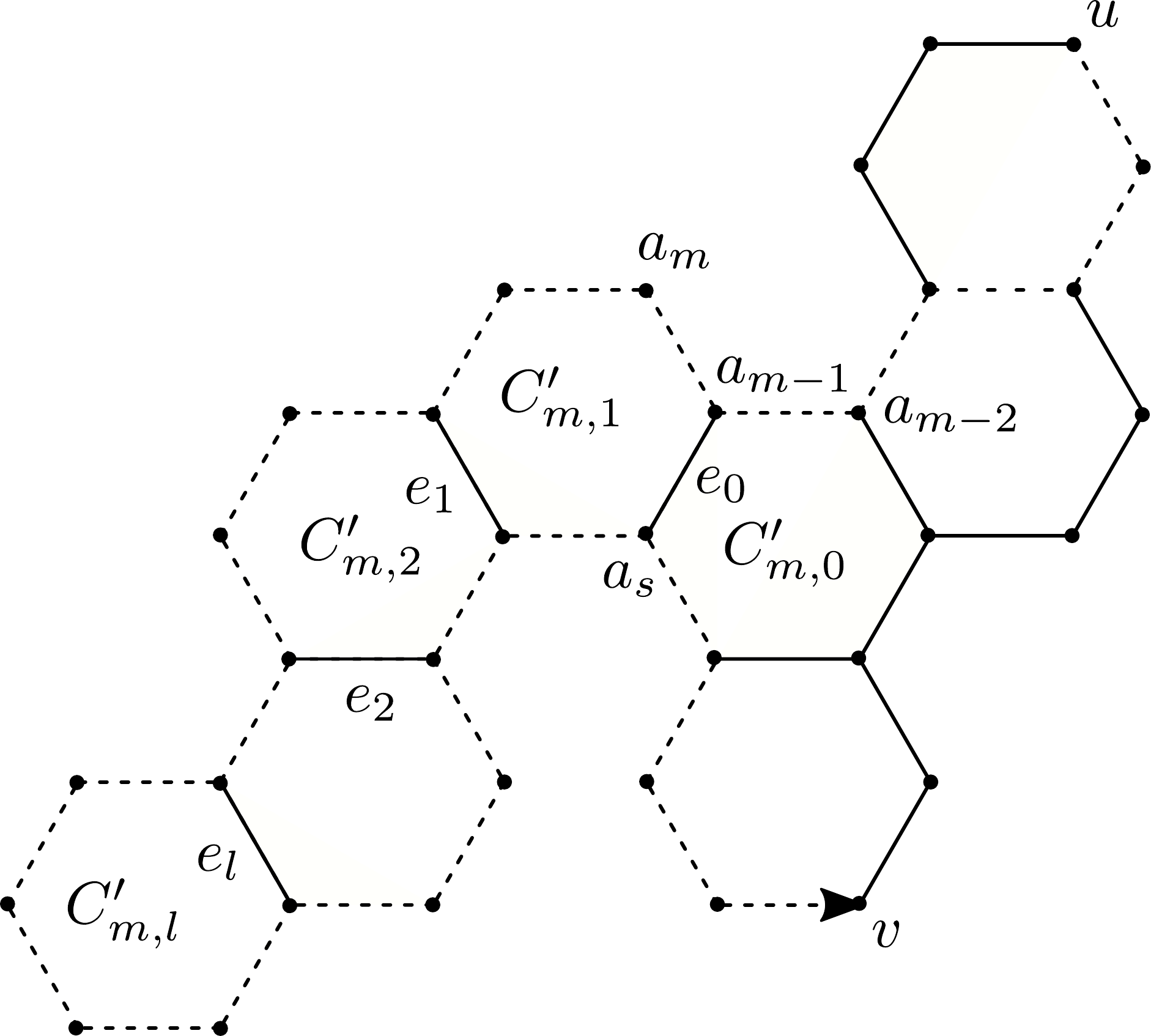}}
\caption{The path in Lemma \ref{lem:monotone_path_down}.}
\label{fig:monotone}
\end{center}
\end{figure}

 \begin{lemma}\label{lem:monotone_path_up}
Let $e=(u,v) \in Q$ where $u \in B_u$ and $v \in B_v$. Suppose that $B_v$ is in the path from $B^*$ to $B_u$ in $\cal{T}$. If $(u,v) \in Q$ then there exists a monotone-up path $P = \{a_1=u,a_2,\ldots,a_t=v\}$ in $G$, where for all $i \in \{1,2,\ldots,t-1\}$ we have $(a_i,a_{i+1}) \in Q$.
\end{lemma}

\begin{proof}
A similar argument as in Lemma \ref{lem:monotone_path_down} applies here.
\end{proof}

\begin{lemma}\label{lem:monotone_up}
Let $e=(u,v) \in Q$ where $u \in B_u$ and $v \in B_v$. Suppose that $B_u$ is in the path from $B^*$ to $B_v$ in $\cal{T}$. If $(u,v) \in Q$ then we have $d(u,v) \leq 3r$.
\end{lemma}

\begin{proof}
Let $P = \{a_1=u,a_2,\ldots,a_t=v\}$ be the monotone-down path obtained by Lemma \ref{lem:monotone_path_down}.
Suppose that the path from $B^*$ to $B_v$ in $\cal{T}$, is a subpath of $P_j$ for some $j \in \{1,2,\ldots,k\}$, and let $p_j$ be the corresponding path in $G$. Let $f = \min_{a_i \in P \cap p_j} i$, and let $l = \max_{a_i \in P \cap p_j} i$. Let $P'=\{a_1,a_2,\ldots,a_f\}$, $P''=\{a_f,a_{f+1},\ldots,a_l\}$, and $P'''=\{a_l,a_{l+1},\ldots,a_t\}$. For any $i \in \{1,2,\ldots,f-1\}$ we have that $(a_i,a_{i+1})$ is not removed after step $4$. This implies that $d(a_1,a_f) \leq r$. With a similar argument, we can show that $d(a_l,a_t) \leq r$. Finally, for any $i \in \{f,f+1,\ldots,l-1\}$ we have that $(a_i,a_{i+1})$ is not removed after step $2$, and thus we have $d(a_f,a_l) \leq r$. Therefore by applying triangle inequality, we get $d(a_1,a_t) \leq 3r$, as required.
\end{proof}

\begin{lemma}\label{lem:monotone_down}
Let $e=(u,v) \in Q$ where $u \in B_u$ and $v \in B_v$. Suppose that $B_v$ is in the path from $B^*$ to $B_u$ in $\cal{T}$. If $(u,v) \in Q$ then we have $d(u,v) \leq 3r$.
\end{lemma}

\begin{proof}
A similar argument as in Lemma \ref{lem:monotone_up} applies here.
\end{proof}

\begin{lemma}\label{lem:rbounded}
If $e=(u,v) \in Q$ then we have $d(u,v) \leq 6r$.
\end{lemma}

\begin{proof}
Since $(u,v) \in Q$, we have that at the beginning of step $8$ there must have been a path $P = \{a_1=u,a_2,\ldots,a_t=v\}$ such that for all $i \in \{1,2,\ldots,t-1\}$ we have $(a_i,a_{i+1}) \in Q$. Suppose that $u \in B_u$ and $v \in B_v$ for some $B_u,B_v \in \cal{X}$. Let $B_w \in \cal{X}$ be a vertex with minimum level that has non-empty intersection with $V(P)$. Let $w \in B_w \cap V(P)$ be an arbitrary vertex. By the construction, we have that $B_w$ is in the paths from $B^*$ to $B_v$ and $B_u$. Moreover, we have $(u,w) \in Q$ and $(w,v) \in Q$. Therefore, by Lemmas \ref{lem:monotone_up} and \ref{lem:monotone_down}, we have that $d(u,w) \leq 3r$ and $d(w,v) \leq 3r$. This implies that $d(u,v) \leq 6r$, as desired.
\end{proof}

\begin{proof}[Proof of Theorem \ref{thm:tw2}]
The proof follows by combining Lemmas \ref{lem:rbounded} and \ref{lem:edge_probtw2}.
\end{proof}

\section{Lipschitz quasipartitions for bounded pathwidth digraphs}\label{sec:pathwidth}
In this Section we provide a proof for Theorem \ref{thm:boundedpathwidth}. We
present an efficient algorithm for computing a random quasipartition of a directed graph with bounded pathwidth.
We begin by describing a specific family of graphs of bounded pathwidth, which we refer to as \emph{path of cliques}.
We show that any graph of pathwidth $k$ admits an isometric embedding into a path of $k$-cliques.
We then describe an algorithm to obtain a random quasipartition of a path of $k$-cliques. Finally, we present the analysis of the aforementioned algorithm.

\subsection{Isometric embedding for bounded pathwidth digraphs into path of cliques}

\paragraph{Path of cliques}
We call a digraph $G$ a \emph{path of $k$-cliques} 
if there exists a collection of pairwise disjoint subsets of $V(G)$ $S_1,S_2, \ldots S_l$ such that the following conditions hold:
\begin{description}
\item{1)} $S_1 \cup S_2 \cup \ldots \cup S_l = V(G)$.
\item{2)} $|S_i| = k$ for any integer $i \in \{ 1,\ldots,l \}$.
\item{3)} $E(G) = \{ (u,v),(v,u)| u,v \in S_i \cup S_{i+1}, i\in \{ 1,\ldots,l-1 \} \}$.
\end{description}

We refer to the subsets $S_1,S_2, \ldots S_l$ as cliques. We call an edge $(u,v) \in E(G)$ a \emph{vertical edge} if $u,v \in S_i$ for some $i\in \{ 1,\ldots,l \}$. We call all other edges in $E(G)$ \emph{horizontal edges}. Furthermore if $u \in S_i$ and $v$ in $S_{i+1}$ for some $i\in \{ 1,\ldots,l \}$ we call the horizontal edge $(u,v)$ \emph{left horizontal}. Otherwise if $v \in S_i$ and $u \in S_{i+1}$ we call the horizontal edge $(u,v)$ \emph{right horizontal}. Let $p$ be a directed path in $G$ from $u$ to $v$ where $u,v \in V(G)$. We define $d_G(p) = d_G(u,v)$. Let $x,y \in V(G)$ be vertices that are traversed by $p$ in that order. Then we denote the sub-path of $p$ from $x$ to $y$ by $p(x,y)$. 
We define the \emph{width} of $p$ to be $\max\limits_{i \in \{ 1,\ldots,l \}}{|p \cap S_i|}$. 

\begin{theorem}
Let $G$ be a weighted digraph of pathwidth $k$. Then $G$ embeds isometrically into a path of $(k+1)$-cliques. 
\end{theorem}

\begin{proof}
Let $(X,P)$ denote a path decomposition of $G$ with pathwidth $k$ where $X = \{X_1 , \ldots X_l \}$ is a collection of subsets of $V(G)$. First we observe that from the definition of pathwidth we have that $|X_i| \leq k+1$ for all $i \in \{ 1,\ldots,l \}$. W.l.o.g we may assume that $|X_i| = k+1$ for all $i$ since otherwise we can add vertices to $X_i$ from either $X_{i+1}$ or $X_{i-1}$ and we know that there exists some $X_j $ such that $|X_j| = k+1$. Now we construct a path of cliques $H$ as follows. We initialize $V(H)$ and $E(H)$ to be $\emptyset$. For all $i$ and for each vertex $u \in V(G)$ in $X_i$ we add to $V(H)$ a unique vertex $u_i$. We call $u$ the parent of $u_i$ in $G$ and denote this by $u = p(u_i)$. Next we define the sets $S_i = {u_i| u \in X_i}$. Clearly $S_1 \cup S_2 \cup \ldots \cup S_l = V(H)$. Next for all $i$ and for all $x,y \in X_i \cup X_{i+1}$ we add the directed edges $(x,y)$ and $(y,x)$ to $E(H)$ setting their weights to be $d_G(p(x),p(y))$ and $d_G(p(y),p(x))$ respectively. Now we are ready to define the isometric mapping $\phi: V(G) \to V(H)$. For each $u \in V(G)$ we pick arbitrarily some $v \in V(H)$ such that $p(v) = u$ and set $\phi(u) = v$. This gives us the desired mapping. Now we observe that every edge in $H$ has weight equal to the shortest path distance between the corresponding parent vertices in $G$. This implies that distances don't contract in the embedding. Consider any directed path $q = \{ a_1, \ldots, a_m \}$ in $G$. For any $u,v \in V(H)$ such that $p(u) = p(v)$ we have that there exist directed paths of length $0$ from $u$ to $v$ and from $v$ to $u$. This follows due to the fact that $H$ is constructed from a path decomposition of $G$. For the same reason it must also be that for every directed edge $(a_i,a_{i+1})$ in $q$ there exist $u,v \in V(H)$ such that $p(u)= a_i$ and $p(v) = a_{i+1}$ and $u,v$ both belong to $X_j$ for some $j$. This implies that there is a corresponding directed edge $(u,v) \in E(H)$. Therefore there is a corresponding directed path in $H$ with the same length as $q$. So it follows that $\phi$ does not dilate distances. Therefore $\phi$ is an isometric embedding.  
\end{proof}

\subsection{Algorithm}\label{alg:boundedpathwidth}
\begin{description}
\item{\textbf{Input:}}
A graph $G$ of pathwidth $k-1$ and a corresponding path of cliques $B_1,B_2,\cdots, B_l$, where for each $i \in \{1,2,\cdots,l\}$ we have $|B_i| = k$.
\item{\textbf{Output:}}
Random quasipartition $Q$ of the shortest-path quasimetric space of $G$.
\item{\textbf{Initialization.}}
Set $G^*=G$, $Q=E(G)$, and $i=1$.
\item{\textbf{Step 1.}}
Pick $z \in [0,r]$ uniformly at random.
\item{\textbf{Step 2.}} 
Pick $v_{1,i} \in B_1$ such that $B_l$ is reachable from $v_{1,i}$ in $G^*$ (There exists a directed path from $v_{1,i}$ to $B_l$ in $G^*$). Let $s_i = v_{1,i}$, and let $t_i \in B_l$ be such that $t_i$ is reachable from $s_i$. Let $p_i$ be a shortest path from $s_i$ to $t_i$ in $G^*$.
\item{\textbf{Step 3.}} 
For all $(u,v) \in E(G^*)$ remove $(u,v)$ from $Q$ if $d_{G^*}(s_i,v)> j\cdot r + z$ and $d_{G^*}(s_i,u) \leq j\cdot r + z$ for some integer $j \geq 0$.
\item{\textbf{Step 4.}}
For any $(u,v) \in p_i$, if $(u,v)$ is horizontal in $G$ then delete $(u,v)$ from $G^*$, and set $i = i+1$. If $i \leq k^2$, go back to Step 2.
\item{\textbf{Step 5.}}
Set $G^* = G$, and $i = 1$.
\item{\textbf{Step 6.}} 
Pick $v_{l,i} \in B_l$ such that $B_1$ is reachable from $v_{l,i}$ in $G^*$ (There exists a directed path from $v_{l,i}$ to $B_1$ in $G^*$). Let $s'_i = v_{l,i}$, and let $t'_i \in B_1$ be such that $t'_i$ is reachable from $s'_i$. Let $q_i$ be a shortest path from $s'_i$ to $t'_i$ in $G^*$.
\item{\textbf{Step 7.}} 
For all $(u,v) \in E(G^*)$ remove $(u,v)$ from $Q$ if $d_{G^*}(s'_i,v)> j\cdot r + z$ and $d_{G^*}(s'_i,u) \leq j\cdot r + z$ for some integer $j \geq 0$.
\item{\textbf{Step 8.}} 
For any $(u,v) \in q_j$, if $(u,v)$ is horizontal in $G$ then delete $(u,v)$ from $G^*$, and set $i = i+1$. If $i \leq k^2$, go back to Step 6.
\item{\textbf{Step 9.}}
Remove every $(u,v) \in E(G)$ from $Q$ with $d(u,v) \geq r$.
\item{\textbf{Step 10.}}
Let $G'$ be the subgraph of $G$ with $V(G') = V(G)$ and $E(G') = Q$.
\item{\textbf{Step 11.}} 
Enforce transitivity on $Q$; that is, for all $u,v,w \in V(G)$ if $(u,v) \in Q$ and $(v,w) \in Q$ then add $(u,w)$ to $Q$.
\end{description}

\subsection{Analysis of the algorithm}
We call an edge in $G'$ \emph{horizontal} (resp. \emph{vertical}) if the corresponding edge in $G$ is \emph{horizontal} (resp. \emph{vertical}).
For all $c \in \{1,\ldots,k\}$,$i \in \{1,\ldots,k^2\}$ and $j \in \{1,\ldots,k^2\}$ we define $P_{c,i,j}$ to be the set of all directed paths in $G'$ with width $c$ that do not traverse any horizontal edge in $p_x$ and in $q_y$ for all $x < i$ and $y < j$ and traverse at least one horizontal edge in $p_i$ and $q_j$.

\begin{lemma}
Let $(u,v) \in Q$  after Step 11 of the algorithm. Then there exists a directed path from $u$ to $v$ in $G'$.
\end{lemma}

\begin{proof}
We have that $(u,v) \in Q$  after Step 11 of the algorithm. Combined with the fact that step 11 enforces transitivity on $Q$ this implies that there are two possibilities. The first case is that $(u,v) \in E(G)$ and $(u,v)$ is not removed from $Q$ by the algorithm in steps $3$ or $7$. The other possibility is that it must have been that there was a sequence of vertices $a_1 =u, a_2, \ldots, a_m = v$ such that for all $i<m$ $(a_i,a_{i+1}) \in E(G)$ and $(a_i,a_{i+1}) \in Q$  after step $9$ of the algorithm. This implies that the corresponding directed edges are present in $G'$. Therefore the directed path $P = \{ a_1,\ldots, a_m \}$ is also present in $G'$.
\end{proof}

\begin{lemma}\label{lem:horizontal_edge}
Let $e \in E(G)$ be a horizontal edge. Then there exists $j \in \{1,2,\ldots,k^2\}$ such that either $e \in p_j$ or $e \in q_j$.
\end{lemma}

\begin{proof}
Let $e = (u,v)$ and w.l.o.g let $u \in B_i$ and $v \in B_{i+1}$ for some $i \in \{1,2,\ldots,l \}$. Since we delete horizontal edges in step 4 and each $p_j$ begins at a vertex in $B_1$ and ends at a vertex in $B_l$ it follows that each $p_j$ contains a unique horizontal edge from $B_i$ to $B_{i+1}$. Combined with the fact that that there are exactly $k^2$ directed edges from $B_i$ to $B_{i+1}$ for any $i$ this implies that for all $i$ every horizontal edge from from $B_i$ to $B_{i+1}$ is contained in some $p_j$. A similar argument can be used to show that for all $i$ every horizontal edge from from $B_i$ to $B_{i-1}$ is contained in some $q_j$. This concludes the proof. 
\end{proof}

\begin{lemma}\label{lem:step_3}
For every $(u,v) \in E(G)$, we have $Pr[(u,v) \text{ is removed from $Q$ by step 3}] \leq k^2 \frac{d(u,v)}{r}$.
\end{lemma}

\begin{proof}
By the construction, an edge $(u,v)$ is removed from $Q$ by step $3$ if $d_{G^*}(s_j,v)> i\cdot r + z$ and $d_{G^*}(s_j,u) \leq i\cdot r + z$ for some integer $i \geq 0$. This means that each time after running this step we have that $(u,v)$ is removed from $Q$ with probabilty at most $\frac{d(u,v)}{r}$.
The algorithm runs this step exactly $k^2$ times, and thus we have $Pr[(u,v) \text{ is removed from $Q$ by step 3}] \leq k^2 \frac{d(u,v)}{r}$, as desired.
\end{proof}

\begin{lemma}\label{lem:step_7}
For every $(u,v) \in E(G)$, we have $Pr[(u,v) \text{ is removed from $Q$ by step 7}] \leq k^2 \frac{d(u,v)}{r}$.
\end{lemma}

\begin{proof}
A similar argument as in Lemma \ref{lem:step3} applies here.
\end{proof}

\begin{lemma}\label{lem:step_9}
For every $(u,v) \in E(G)$, we have $Pr[(u,v) \text{ is removed from $Q$ by step 9}] \leq \frac{d(u,v)}{r}$.
\end{lemma}

\begin{proof}
If $d(u,v) < r$, then $(u,v)$ is not removed by step 9. Otherwise, $(u,v)$ is removed from $Q$ and we have $Pr[(u,v) \text{ is removed from $Q$ by step 9}] = 1 \leq \frac{d(u,v)}{r}$.
\end{proof}

\begin{lemma}\label{lem:probabilitybound}
For every $(u,v) \in E(G)$, we have $Pr[(u,v) \text{ is removed from $Q$}] \leq (2k^2+1) \frac{d(u,v)}{r}$.
\end{lemma}

\begin{proof}
This follows immediately by Lemmas \ref{lem:step_3}, \ref{lem:step_7}, and \ref{lem:step_9}.
\end{proof}

Now we show that $Q$ is $2^{k^{O(1)}}r$-bounded.

\begin{lemma}\label{lem:base_case_1}
Let $u,v \in V(G)$ and $p \in P_{c,k^2+1,k^2+1}$ be a directed path in $G'$ from $u$ to $v$. Then $d_G(u,v) \leq (c-1)r$.
\end{lemma}

\begin{proof}
We have that $p$ does not traverse any horizontal edges in $G'$ that are in $p_1,\ldots,p_{k^2}$ or $q_1,\ldots,q_{k^2}$. But since all horizontal edges in $G'$ are either part of $p_1,\ldots,p_{k^2}$ or part of $q_1,\ldots,q_{k^2}$ this implies that $p$ consists only of vertical edges of which there can be at most $c-1$ since each clique consists of at most $k$ vertices of which $p$ can visit at most $c$ since $p$ has width $c$. Since each edge in $G'$ has length at most $r$ from step 9 of the algorithm it follows that $d_G(p) \leq (c-1)r$.
\end{proof}

\begin{lemma}\label{lem:distance_bound}
Let $u,v \in V(G)$ and $i \in \{1, \ldots, k^2 \}$ such that $u$ and $v$ are traversed by $p_i$ in that order. Suppose there exists a path $p$ from $u$ to $v$ that does not traverse any horizontal edge of $p_j$ for all $j < i$ then $d_G(u,v) \leq r$.
\end{lemma}

\begin{proof}
The above lemma can be proved by considering the $i$th iteration of step 3 of the algorithm. Since $p$ does not traverse any horizontal edge of $p_j$ for all $j < i$ it follows that $p$ is a path in $G^*$. However since no edge of $p$ is removed in step $3$ this implies that $d_{G^*}(s_i,v) - d_{G^*}(s_i,u) \leq r$. This combined with the fact that $p_i$ is a shortest path in $G^*$ implies that $d_G(u,v) \leq r$.  
\end{proof}

\begin{figure}[h]
\begin{center}
\scalebox{0.15}{\includegraphics{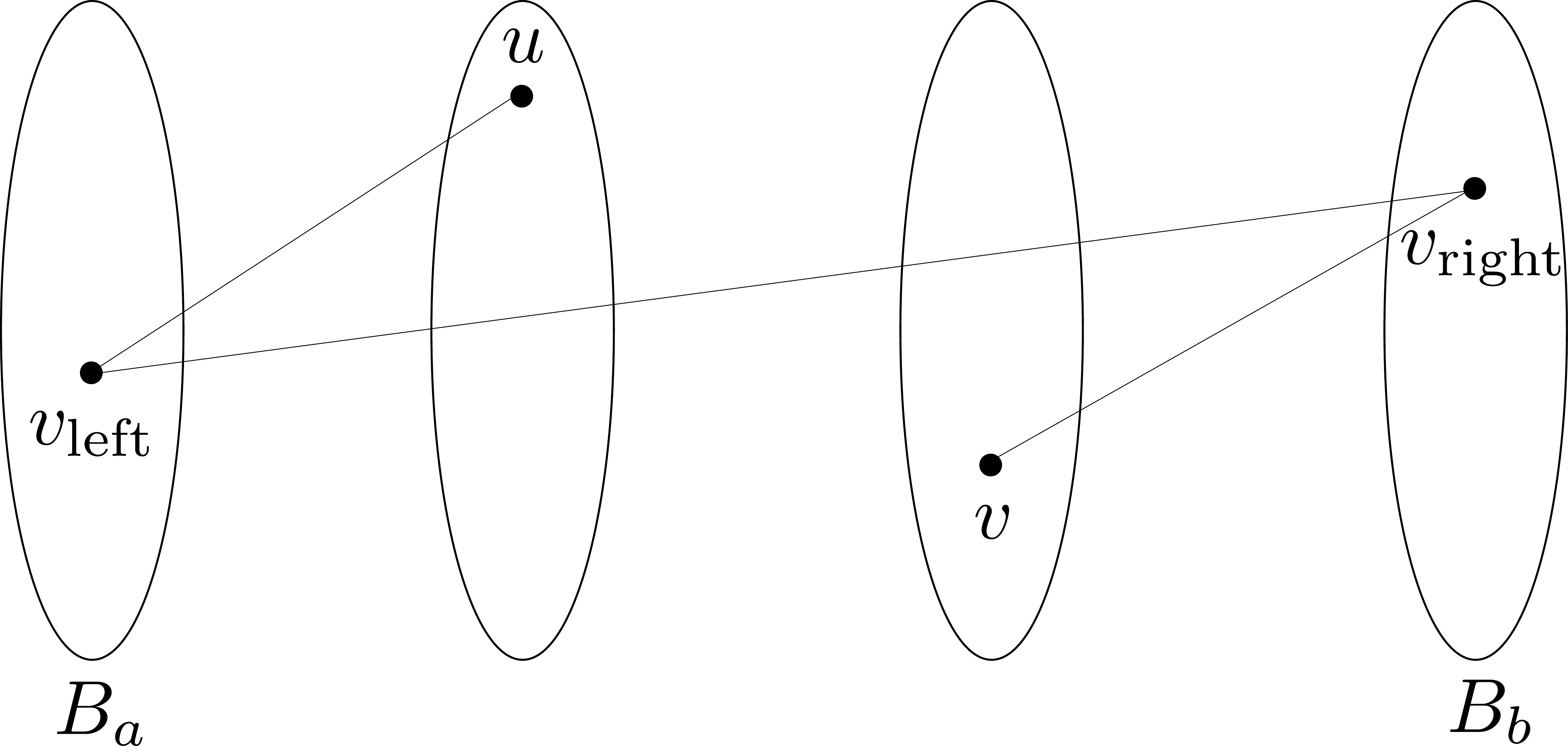}}
\caption{The path $p$ in Lemma \ref{lem:inductive_step}.}
\label{fig:leftright}
\end{center}
\end{figure}

\begin{lemma}\label{lem:inductive_step}
Let $u,v \in V(G)$ and $p \in P_{c,i,j}$ be a directed path in $G'$ from $u$ to $v$. Then $d_G(u,v) \leq c^{O(2k^2 - i - j)}r$.
\end{lemma}

\begin{proof}
We prove this by induction. Lemma \ref{lem:base_case_1} will be used as the base case. Suppose the statement is true for directed paths in $P_{d,i,j}$,$P_{c,z,j}$ and $P_{c,i,o}$ where $d<c$, $z>i$, and $o>j$. We denote by $v_{\lef}$ and $v_{\ri}$ a pair of leftmost and rightmost vertices traversed by $p$. More precisely we  pick a vertex $v_{\lef}$ from $B_a$ where $a$ is the smallest integer such that $B_a \cap p \neq \emptyset$. Similarly we  pick a vertex $v_{\ri}$ from $B_b$ where $b$ is the largest integer such that $B_b \cap p \neq \emptyset$. Suppose $p$ traverses $B_a$ before $B_b$, then we pick $v_{\lef}$ to be the first vertex in $B_a$ visited by $p$ and we pick $v_{\ri}$ to be the last vertex in $B_b$ visited by $p$. Otherwise we pick $v_{\lef}$ to be the last vertex in $B_a$ visited by $p$ and we pick $v_{\ri}$ to be the first vertex in $B_b$ visited by $p$. Note that if $c=1$, then either $v_{\lef} = u$ and $v_{\ri} = v$ or $v_{\lef} = v$ and $v_{\ri} = u$. Consider the case that $p$ traverses $v_{\lef}$ before $v_{\ri}$ (See figure \ref{fig:leftright}). Then we have that $p= p(u,v_{\lef}) \cup p(v_{\lef},v_{\ri}) \cup p(v_{\ri},v)$. Next we observe that since the width of $p$ is $c$ and the width of $p(v_{\lef},v_{\ri})$ is at least $1$ it follows that $p(u,v_{\lef})$ and $p(v_{\ri},v)$ have width at most $c-1$. Consider $p(v_{\lef},v_{\ri})$. Let $u_1$ be the first vertex in $p(v_{\lef},v_{\ri})$ that intersects $p_i$. Let $v_1$ be the last vertex in $p(v_{\lef},v_{\ri})$ that intersects $p_i$ beyond $u_1$. Suppose $p(v_1,v_{\ri})$ does not traverse any horizontal edges of $p_i$ then we have that $p(v_{\lef},u_1)$ and $p(v_1,v_{\ri})$ are in $P_{d,z,o}$  for some integers $d \leq c$,$z > i$ and $o \geq j$. We also have that $p(u,v_{\lef}),p(v_{\ri},v) \in P_{d,z,o}$ for some integers $d<c$,$z \geq i$ and $o \geq j$ since $p(v_{\lef},v_{\ri})$ has width at least $1$. So we can use the distance bounds from the induction on $p(v_{\lef},u_1)$, $p(v_1,v_{\ri})$, $p(u,v_{\lef})$, and $p(v_{\ri},v)$. For $p(u_1,v_1)$ we may bound the distance using lemma \ref{lem:distance_bound}. This gives us the required bound. Suppose $p(v_1,v_{\ri})$ traverses some horizontal edge of $p_i$ then let $u_2$ be the first vertex in $p(v_{\lef},v_{\ri})$ traversed after $v_1$ that intersects $p_i$. By the choice of $v_1$ it must be that $u_2$ is a vertex traversed before $u_1$ by $p_i$. Let $u_1$ belong to a clique $B_t$ for some $t \in \{1,\ldots,l\}$. Now we have that $p$ traverses some vertex $x_1 \in B_t$ after visiting $v_1$ and before visiting $u_2$ since $u_2$ is in a clique with a lower index than $u_1$. Let $y_1 \in B_t$ be the last vertex in $B_t$ traversed by $p$. Now it follows that $p(x_1,y_1)$ has width at most $c-1$. This is because $p(x_1,y_1)$ does not intersect $p(v_{\lef}),u_1)$ and does not intersect $p(y_1,v_{\ri})$ except at $y_1$. From our choice of $v_1$ it must be that $p(y_1,v_{\ri})$ does not intersect any horizontal edge of $p_i$. Similarly $p(v_1,x_1)$ does not intersect any horizontal edge of $p_i$ from our choice of $u_2$ and $x_1$. Thus again we can bound the distances of $p(v_{\lef},u_1)$, $p(v_1,x_1)$, $p(x_1,y_1)$, $p(y_1,v_{\ri})$, $p(u,v_{\lef})$, and $p(v_{\ri},v)$ using induction and bound the distance of $p(u_1,v_1)$ using lemma \ref{lem:distance_bound} giving the required bound on the distance. For the case that $p$ traverses $v_{\ri}$ before $v_{\lef}$ we use a similar argument where we consider $q_j$ instead of $p_i$.
\end{proof}

\begin{figure}
\begin{center}
\scalebox{0.15}{\includegraphics{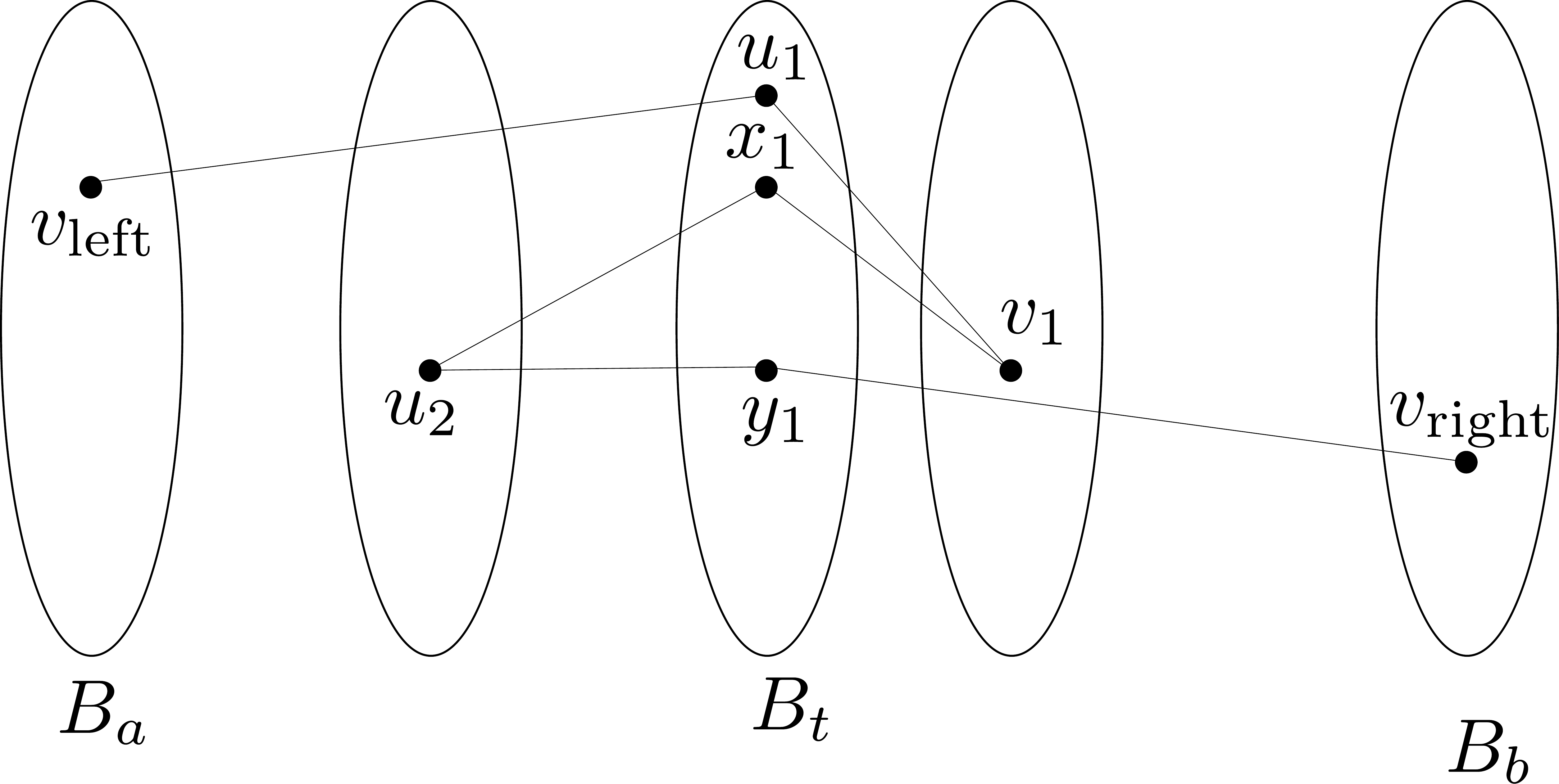}}
\caption{The sub-path of $p$ in Lemma \ref{lem:inductive_step}.}
\label{fig:induction}
\end{center}
\end{figure}

\begin{proof}[Proof of Theorem \ref{thm:boundedpathwidth}]
The required quasipartition is obtained from Algorithm \ref{alg:boundedpathwidth} setting $r = \frac{\Delta}{2^{\alpha k^2}}$ where $\alpha$ is the constant in the asymptotic notation from Lemma \ref{lem:inductive_step}. The proof follows from Lemmas \ref{lem:inductive_step} and \ref{lem:probabilitybound}.
\end{proof}

\subsection{Derandomization}
It is possible to derandomize the algorithm in subsection \ref{alg:boundedpathwidth} to get a polynomial time algorithm that enumerates the support of the distribution in Theorem \ref{thm:boundedpathwidth}. The approach is quite similar to that used in Section 3 Theorem 14 of \cite{memoli2016quasimetric}. We observe that there can be at most $n-1$ unique combinations of ordered pair removals in each iteration of Step 3 of the algorithm. In order to enumerate these we only need to consider values of $z$ for which there exists $u \in V(G^*)$ such that $d_{G^*}(s_i,u) = j\cdot r +z$. Using a similar argument for each iteration of Step 3 we only need to consider at most  $(n-1)k^2$ unique values of $z$ in total to enumerate the support of the distribution obtained in the algorithm.

\begin{theorem}\label{thm:boundedpathwidthalgo}
Let $G$ be a directed graph of of pathwidth $k$.
Let $M = (V(G),d_G)$ denote the shortest-path quasimetric space of $G$. Then for all $\Delta > 0$, there exists a polynomial time algorithm with running time of the form $n^{O(1)}$ that can compute the support of some $2^{O(k^2)}$-Lipschitz distribution over $\Delta$-bounded quasipartitions of $M$.
\end{theorem}

\section{Embedding directed cycles into directed $\ell_1$}\label{sec:cyclel1}

In this section we describe an algorithm for computing a constant-distortion embedding of the shortest-path quasimetric space of directed cycles into a convex combination of 0-1 quasimetric spaces.

We begin by introducing some notation.
Consider a directed cycle $G$ and fix a planar drawing of $G$ where $G$ is drawn as a circle centered at the origin.
Let $w$ be a weight function on the edges of $G$. 
Let $X = (V,d)$ be the shortest-path quasimetric space of $G$.
We may assume w.l.o.g.~that every edge $(u,v) \in E$ is the shortest path from $u$ to $v$ in $G$.
We denote by $\Delta$ the diameter of $G$.
For all $u,v \in V$ we denote by $\overleftarrow{p}(u,v)$ the path from $u$ to $v$ in the counter-clockwise direction in the drawing of $G$ and by $\overleftarrow{d}(u,v)$ the length of $\overleftarrow{p}(u,v)$. Similarly we denote by $\overrightarrow{p}(u,v)$ the path from $u$ to $v$ in the clockwise direction in the drawing of $G$ and by $\overrightarrow{d}(u,v)$ the length of $\overrightarrow{p}(u,v)$.
We say that some $v\in V(G)$ is the \emph{meet-point} of $u$ if $\overleftarrow{d}(u,v) = \overrightarrow{d}(u,v)$.
For some $(u,v)\in V(G)\times V(G)$, we say that
$(u,v)$ is \emph{short} if $\overleftarrow{d}(u,v) <\frac{\Delta}{10}$ and $\overrightarrow{d}(u,v) <\frac{\Delta}{10}$; otherwise we that $(u,v)$ is  \emph{long}.
Let $A \subseteq V$ where $A = \{ u : u \in V, \exists v \in V  \text{ such that } \overleftarrow{d}(u,v) < \frac{\Delta}{10} \text{ and } \overrightarrow{d}(u,v) < \frac{\Delta}{10} \}$. Similarly let $B \subseteq V$ where  $B = \{ v : v \in V, \exists u \in V  \text{ such that } \overleftarrow{d}(u,v) < \frac{\Delta}{10} \text{ and } \overrightarrow{d}(u,v) < \frac{\Delta}{10} \}$.

\begin{figure}[h]
\begin{center}
\scalebox{0.29}{\includegraphics{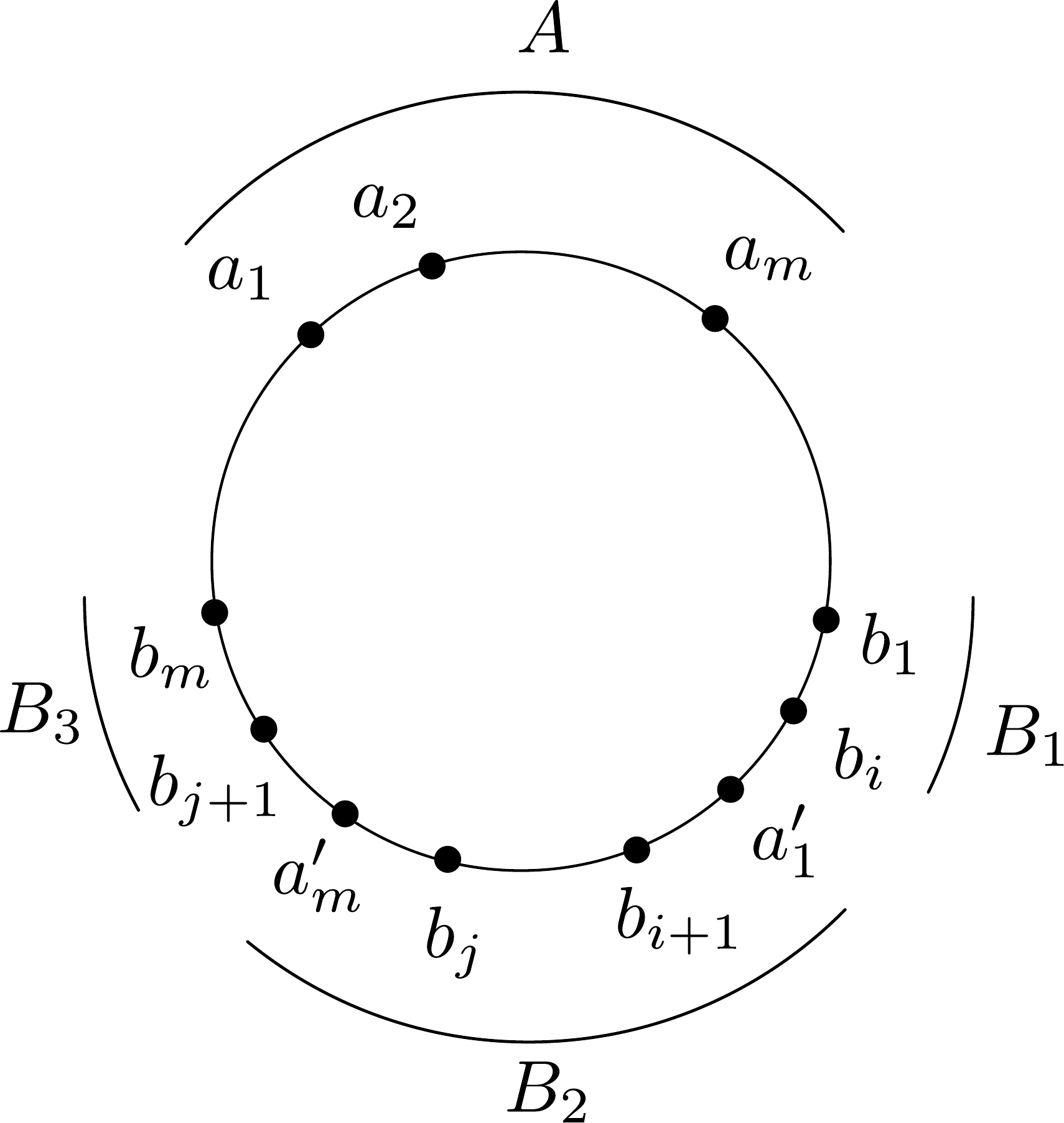}}
\caption{The sets $A$ and $B$.}
\label{fig:cycle}
\end{center}
\end{figure}

\begin{lemma}\label{lem:cycle_structure}
Let $u_1,u_2,u_3\in V(G)$ with 
$\overrightarrow{p}(u_1,u_2) \cap \overrightarrow{p}(u_2,u_3) = u_2$.
Let $(u_1,v_1), (u_2,v_2), (u_3,v_3)$ be short pairs.
Let $A' = \{u_1,u_2,u_3\}$ and $B' = \{v_1,v_2,v_3\}$.
Then we have that all the vertices in $B'$ are contained within one of $\overrightarrow{p}(u_1,u_2), \overrightarrow{p}(u_2,u_3), \overrightarrow{p}(u_3,u_1)$ . 
\end{lemma}

\begin{proof}
Suppose for the sake of contradiction that the assertion does not hold.
W.l.o.g.~assume that $\overrightarrow{p}(u_1,u_2)$ and $\overrightarrow{p}(u_2,u_3)$ each intersect $B'$. We consider the following cases:

Case 1: $\overrightarrow{p}(u_1,u_2)$ contains $v_1$ and $\overrightarrow{p}(u_2,u_3)$ contains $v_2$. In this case we have that $\overleftarrow{d}(u_1,v_1) < \frac{\Delta}{10}$ and $\overleftarrow{d}(u_2,v_2) < \frac{\Delta}{10}$. But $\overleftarrow{p}(u_1,v_1) \cup \overleftarrow{p}(u_2,v_2)$ contains all the edges of $G$ in the counter-clockwise direction. This implies that the sum of the lengths of all the edges in the counter-clockwise direction  is at most $\frac{2 \Delta}{10}$. However this means that the diameter of $G$ is at most $\frac{2 \Delta}{10}$ which is a contradiction.

Case 2: $\overrightarrow{p}(u_1,u_2)$ contains $v_2$ and $\overrightarrow{p}(u_2,u_3)$ contains $v_1$. In this case we have that $\overrightarrow{d}(u_1,v_1) < \frac{\Delta}{10}$ and $\overrightarrow{d}(u_2,v_2) < \frac{\Delta}{10}$. But $\overrightarrow{p}(u_1,v_1) \cup \overrightarrow{p}(u_2,v_2)$ contains all the edges of $G$ in the clockwise direction. This implies that the sum of the lengths of all the edges in the clockwise direction  is at most $\frac{2 \Delta}{10}$. However this means that the diameter of $G$ is at most $\frac{2 \Delta}{10}$ which is a contradiction.

Case 3: $\overrightarrow{p}(u_1,u_2)$ contains $v_3$, $\overrightarrow{p}(u_2,u_3)$ contains $v_1$ and $\overrightarrow{p}(u_3,u_1)$ contains $v_2$. In this case we have that $\overrightarrow{d}(u_1,v_1) < \frac{\Delta}{10}$, $\overrightarrow{d}(u_2,v_2) < \frac{\Delta}{10}$ and $\overrightarrow{d}(u_3,v_3) < \frac{\Delta}{10}$. But $\overrightarrow{p}(u_1,v_1) \cup \overrightarrow{p}(u_2,v_2) \cup \overrightarrow{p}(u_3,v_3)$ contains all the edges of $G$ in the clockwise direction. This implies that the sum of the lengths of all the edges in the clockwise direction  is at most $\frac{3 \Delta}{10}$. However this means that the diameter of $G$ is at most $\frac{3 \Delta}{10}$ which is a contradiction.

Case 4: $\overrightarrow{p}(u_1,u_2)$ contains $v_2$ and $\overrightarrow{p}(u_3,u_1)$ contains $v_1$. In this case we have that $\overrightarrow{d}(u_1,v_1) < \frac{\Delta}{10}$ and $\overrightarrow{d}(u_2,v_2) < \frac{\Delta}{10}$. But $\overrightarrow{p}(u_1,v_1) \cup \overrightarrow{p}(u_2,v_2)$ contains all the edges of $G$ in the clockwise direction. This implies that the sum of the lengths of all the edges in the clockwise direction  is at most $\frac{2 \Delta}{10}$. However this means that the diameter of $G$ is at most $\frac{2 \Delta}{10}$ which is a contradiction.
   
Case 5: $\overrightarrow{p}(u_1,u_2)$ contains $v_1$ and $\overrightarrow{p}(u_3,u_1)$ contains $v_2$. In this case we have that $\overleftarrow{d}(u_1,v_1) < \frac{\Delta}{10}$ and $\overleftarrow{d}(u_2,v_2) < \frac{\Delta}{10}$. But $\overleftarrow{p}(u_1,v_1) \cup \overleftarrow{p}(u_2,v_2)$ contains all the edges of $G$ in the counter-clockwise direction. This implies that the sum of the lengths of all the edges in the counter-clockwise direction  is at most $\frac{2 \Delta}{10}$. However this means that the diameter of $G$ is at most $\frac{2 \Delta}{10}$ which is a contradiction.

All other cases are identical to the above cases up to relabeling of the vertices. Since we end up with a contradiction in every case, this concludes the proof.   
\end{proof}

Lemma \ref{lem:cycle_structure} implies that $A$ and $B$ are two disjoint consecutive segments of $G$. Let $a_1,a_2, \ldots, a_m$ be the vertices of $A$ in clockwise order.
Let also $b_1, b_2, \ldots, b_k$ be the vertices of $B$ in clockwise order.
Let $a'_1$ and $a'_m$ be the meet-points of $a_1$ and $a_m$ respectively. If $a_1$ and $a_m$ do not have such points we may introduce two vertices $a'_1$ and $a'_m$ by subdividing edges such that the induced quasimetric does not change. By the construction, we have $d(a_1,a'_1) < \frac{\Delta}{10}$ and $d(a_m,a'_m) < \frac{\Delta}{10}$. Therefore, $a'_1$ lies between $b_i$ and $b_{i+1}$ for some $i \in \{1,2,\ldots,k\}$. Similarly, $a'_m$ lies between $b_j$ and $b_{j+1}$ for some $j \in \{1,2,\ldots,k\}$ (See Figure \ref{fig:cycle}). Let $B_1 = \{b_1,b_2,\ldots,b_i\}$, $B_2 = \{a'_1,b_{i+1},\ldots,a'_m\}$, and $B_3 = \{a_{j+1},\ldots,b_k\}$.

\subsection{The algorithm for embedding into directed $\ell_1$}

We are now ready to describe the algorithm for computing a random quasipartition of $G$.

\begin{flushleft}
\textbf{Input:} A directed cycle $G$.\\
\textbf{Output:} A random quasipartition $Q$ of the shortest-path quasimetric space of $G$.
\end{flushleft}
\begin{description}
\item{\textbf{Initialization:}}
Set $Q=E(G)$.

\item{\textbf{Step 1.}} Pick an arbitrary $v_1 \in V$. Let $v_1, v_2, \ldots, v_n$ be the vertices of $G$ in clockwise order. Let $v_{n+1} = v_1$. Pick $z_1,z_2 \in [0,\frac{\Delta}{10}]$ uniformly at random.

\item{\textbf{Step 2.}} For every $j \in \{1,2, \ldots, n\}$, remove $(v_j,v_{j+1})$ from $Q$ if $\overrightarrow{d}(v_1,v_j) \leq i \cdot \frac{\Delta}{10} + z_1$ and $\overrightarrow{d}(v_1,v_{j+1}) > i \cdot \frac{\Delta}{10} + z_1$ for some integer $i \geq 0$.

\item{\textbf{Step 3.}} For every $j \in \{1,2, \ldots, n\}$, remove $(v_{j+1},v_j)$ from $Q$ if $\overleftarrow{d}(v_1,v_{j+1}) \leq i \cdot \frac{\Delta}{10} + z_1$ and $\overleftarrow{d}(v_1,v_j) > i \cdot \frac{\Delta}{10} + z_1$ for some integer $i \geq 0$.

\item{\textbf{Step 4.}} Pick $z_3 \in [0,\Delta]$ uniformly at random.

\item{\textbf{Step 5.}} For every $(u,v) \in E(\overrightarrow{p}(a_1,a'_1))$, remove $(u,v)$ from $Q$ if $\overrightarrow{d}(a_1,u) \leq z_3$ and $\overrightarrow{d}(a_1,v) > z_3$.

\item{\textbf{Step 6.}} For every $(u,v) \in E(\overleftarrow{p}(a_1,a'_1))$, remove $(u,v)$ from $Q$ if $\overleftarrow{d}(a_1,u) \leq z_3$ and $\overleftarrow{d}(a_1,v) > z_3$.

\item{\textbf{Step 7.}} For every $(u,v) \in E(\overrightarrow{p}(a_m,a'_m))$, remove $(u,v)$ from $Q$ if $\overrightarrow{d}(a_m,u) \leq z_3$ and $\overrightarrow{d}(a_m,v) > z_3$.

\item{\textbf{Step 8.}} For every $(u,v) \in E(\overleftarrow{p}(a_m,a'_m))$, remove $(u,v)$ from $Q$ if $\overleftarrow{d}(a_m,u) \leq z_3$ and $\overleftarrow{d}(a_m,v) > z_3$.

\item{\textbf{Step 9.}} For every $(u,v) \in E(\overrightarrow{p}(a_1,a'_1))$, remove $(u,v)$ from $Q$ if $\overrightarrow{d}(v,a'_1) \leq z_3$ and $\overrightarrow{d}(u,a'_1) > z_3$.

\item{\textbf{Step 10.}} For every $(u,v) \in E(\overleftarrow{p}(a_1,a'_1))$, remove $(u,v)$ from $Q$ if $\overleftarrow{d}(v,a'_1) \leq z_3$ and $\overleftarrow{d}(u,a'_1) > z_3$.

\item{\textbf{Step 11.}} For every $(u,v) \in E(\overrightarrow{p}(a_m,a'_m))$, remove $(u,v)$ from $Q$ if $\overrightarrow{d}(v,a'_m) \leq z_3$ and $\overrightarrow{d}(u,a'_m) > z_3$.

\item{\textbf{Step 12.}} For every $(u,v) \in E(\overleftarrow{p}(a_m,a'_m))$, remove $(u,v)$ from $Q$ if $\overleftarrow{d}(v,a'_m) \leq z_3$ and $\overleftarrow{d}(u,a'_m) > z_3$.

\item{\textbf{Step 13.}} Enforce transitivity on $Q$. That is for all $u,v,w \in V(G)$ if $(u,v) \in Q$ and $(v,w) \in Q$, then add $(u,w)$ to $Q$.
\end{description}
This concludes the description of the algorithm.

\subsection{Analysis}

We now analyze the random quasipartition computed by the above algorithm.
We first argue that the ``probability of separation'' is small for all pairs of vertices; this implies that the contraction of the resulting embedding into $\ell_1$ is bounded.

\begin{lemma}\label{lem:lower prob}
For every $u,v \in V$ we have
\[
\Pr_{Q \sim D}[(u,v) \notin Q] \geq \frac{d(u,v)}{2\Delta}.
\]
\end{lemma}

\begin{proof}
Let $X_1$ be the random event that there exists an edge $e \in \overrightarrow{p}(u,v)$ such that $e \notin Q$. Let $X_2$ be the random event that there exists an edge $e \in \overleftarrow{p}(u,v)$ such that $e \notin Q$. We want to show that $\Pr_{Q \sim D}[X_1 \land X_2] \geq \frac{d(u,v)}{\Delta}$.

First suppose that $(u,v)$ is a long pair. That is, either $\overleftarrow{d}(u,v) \geq \frac{\Delta}{10}$ or $\overrightarrow{d}(u,v) \geq \frac{\Delta}{10}$. We may assume w.l.o.g.~that $\overrightarrow{d}(u,v) \geq \frac{\Delta}{10}$ (the analysis for the other case is similar).
In this case, by the Step 2 of the algorithm we have that $\Pr_{Q \sim D}[X_1] = 1$. This is because of the fact that $\overrightarrow{d}(u,v) \geq \frac{\Delta}{10}$.
Also, from Step 3 of the algorithm we get $\Pr_{Q \sim D}[X_2] \geq \frac{d(u,v)}{\Delta}$. Therefore, we have $\Pr_{Q \sim D}[X_1 \land X_2] \geq \frac{d(u,v)}{\Delta}$.

Now suppose that $(u,v)$ is a short pair. First suppose that $v \in B_1$. In this case, by the construction we have that $d(u,v) = \overrightarrow{d}(u,v)$, and thus by Step $5$ of the algorithm we have that $\Pr_{Q \sim D}[X_1] \geq \frac{d(u,v)}{\Delta}$. Also, by Steps $5$ and $6$ we have $\Pr_{Q \sim D}[X_2|X_1] = 1$, and hence we have $\Pr_{Q \sim D}[X_1 \land X_2] \geq \frac{d(u,v)}{\Delta}$, as desired. If $v \in B_3$, a similar argument using Steps $7$ and $8$ of the algorithm shows that $\Pr_{Q \sim D}[X_1 \land X_2] \geq \frac{d(u,v)}{\Delta}$.

Now suppose that $v \in B_2$. Let $p_1 = \overrightarrow{p}(u,a'_1)$, $p_2 = \overrightarrow{p}(a_m,v)$, $p_3 = \overleftarrow{p}(u,a'_m)$, and $p_4 = \overleftarrow{p}(a_1,v)$.
For every $i \in \{1,2,3,4\}$, let $Y_i$ be the random event that there exists an edge $e \in p_i$ such that $e \notin Q$.
There are four cases:
\begin{description}
\item{Case $1$.} $\overrightarrow{d}(u,a'_1) \geq d(u,v)/2$ and $\overleftarrow{d}(u,a'_m) \geq d(u,v)/2$. First assume that $\overrightarrow{d}(u,a'_1) \leq \overleftarrow{d}(u,a'_m)$. By Step $6$ of the algorithm we have that $\Pr_{Q \sim D}[X_1] \geq \Pr_{Q \sim D}[Y_1] \geq \frac{\overrightarrow{d}(u,a'_1)}{\Delta} \geq \frac{d(u,v)}{2\Delta}$. Also by the construction and Step $10$ of the algorithm, we have that $\Pr_{Q \sim D}[X_2|Y_1] = 1$. Therefore, we have $\Pr_{Q \sim D}[X_1 \land X_2] \geq \frac{d(u,v)}{2\Delta}$, as desired. The argument for the other case where $\overrightarrow{d}(u,a'_1) > \overleftarrow{d}(u,a'_m)$ is similar by considering Steps $8$ and $12$ of the algorithm.

\item{Case $2$.} $\overrightarrow{d}(u,a'_1) \geq d(u,v)/2$ and $\overleftarrow{d}(a_1,v) \geq d(u,v)/2$.
If $\overrightarrow{d}(u,a'_1) \leq \overleftarrow{d}(a_1,v)$, then we follow a similar argument as in the first case by considering Steps $9$ and $10$ of the algorithm. Otherwise, we follow a similar argument by considering Steps $6$ and $10$.

\item{Case $3$.} $\overrightarrow{d}(a_m,v) \geq d(u,v)/2$ and $\overleftarrow{d}(u,a'_m) \geq d(u,v)/2$. If $\overrightarrow{d}(a_m,v) \leq \overleftarrow{d}(u,a'_m)$ then a similar argument as in the first case by considering Steps $7$ and $8$ applies here. Otherwise, a similar argument considering Steps $8$ and $12$ applies here.

\item{Case $4$.} $\overrightarrow{d}(a_m,v) \geq d(u,v)/2$ and $\overleftarrow{d}(a_1,v) \geq d(u,v)/2$.
If $\overrightarrow{d}(a_m,v) \leq \overleftarrow{d}(a_1,v)$, the a similar argument using Steps $7$ and $11$ of the algorithm applies here. Otherwise, a similar argument using Steps $5$ and $6$ applies here.
\end{description}
By considering all the above cases, we conclude that $\Pr_{Q \sim D}[(u,v) \notin Q] \geq \frac{d(u,v)}{2\Delta}$, as desired.
\end{proof}

Next we prove that the probability of separation is not too large; this implies that the expansion of the resulting embedding into directed $\ell_1$ is bounded.

\begin{lemma}\label{lem:edge_prob}
For every $(u,v) \in E$ we have
\[
\Pr_{Q \sim D}[(u,v) \notin Q] \leq \frac{14 d(u,v)}{\Delta}.
\]
\end{lemma}

\begin{proof}
Consider any edge $(u,v) \in E$. First we observe that Steps 2,5,7,9 and 11 only remove edges in the clockwise direction. Similarly Steps 3,6,8,10 and 12 only remove edges in the counter-clockwise direction. W.l.o.g.~let us assume that $(u,v)$ is in the clockwise direction. This means that it may be removed in Step 2 with probability at most $\frac{w(u,v)}{\Delta/10} = \frac{10d(u,v)}{\Delta}$. Now in each of Steps 5,7,9 and 11 we remove $(u,v)$ with probability at most $\frac{w(u,v)}{\Delta} = \frac{d(u,v)}{\Delta}$. Taking the union bound we have that the probability that $(u,v)$ is removed from $Q$ is at most $\frac{14d(u,v)}{\Delta}$.
\end{proof}

\begin{lemma}\label{lem:upper prob}
For every $u,v \in V$ we have
\[
\Pr_{Q \sim D}[(u,v) \notin Q] \leq \frac{14 d(u,v)}{\Delta}.
\]
\end{lemma}

\begin{proof}
Let $p$ be a shortest path from $u$ to $v$ in $G$. Suppose that $(u,v) \notin Q$. Since we enforce transitivity in Step 13 this implies that an edge in $p$ is removed before Step 13. So we have that
$\Pr_{Q \sim D}[(u,v) \notin Q] \leq \ Pr[\text{An edge in $p$ is removed before Step 13}] \leq \sum\limits_{(u,v) \in p} \Pr_{Q \sim D}[(u,v) \notin Q]$ using the union bound. Now applying Lemma \ref{lem:edge_prob} to all the edges of $p$ we have that
\[
\Pr_{Q \sim D}[(u,v) \notin Q] \leq \frac{14 d(u,v)}{\Delta},
\]
which concludes the proof.
\end{proof}

We are now ready to prove the main result of this Section.

\begin{theorem}\label{thm:cycle_01}
Let $G$ be a directed cycle and let $M=(V(G), d_G)$ be its shortest-path quasimetric space.
Then $M$ admits a constant-distortion embedding into some convex combination of 0-1 quasimetric spaces denoted by $D$.
Moreover we can sample a random 0-1 quasimetric space from $D$ in polynomial time.
\end{theorem}

\begin{proof}
The required convex combination of 0-1 quasimetrics is obtained from the distribution $D$ returned by the Algorithm. Every quasipartition $Q$ in the support of $D$ can be replaced with a 0-1 quasimetric where for any $u,v \in V$ $d(u,v) = 0$ iff $(u,v) \in Q$ otherwise $d(u,v) =0$. Since $D$ is a probability distribution over quasipartitions this gives us a convex combination of 0-1 quasimetrics $\phi$. Now it remains to show that the distortion is $O(1)$.

For any $u,v \in V$ we denote the distance from $u$ to $v$ in $\phi$ $d_{\phi}(u,v) $. Now we have that ,
\[ 
d_{\phi}(u,v) = 1\cdot Pr_{Q \sim D}[(u,v) \notin Q]  + 0\cdot Pr_{Q \sim D}[(u,v) \in Q] =Pr_{Q \sim D}[(u,v) \notin Q] 
\]

From Lemmas \ref{lem:lower prob} and \ref{lem:upper prob} we have that,

\[ 
d(u,v)\cdot \frac{1}{2\Delta} \leq d_{\phi}(u,v) \leq 28 d(u,v) \cdot \frac{1}{2\Delta}
\]
This implies that the distortion is at most 28.
The bound on the running time is immediate from the description of the algorithm.
\end{proof}

By using Theorem \ref{thm:cycle_01}, we can obtain the following.

\begin{proof}[Proof of Theorem \ref{thm:cycle_ell1}]
First we observe that any quasipartition in the support of the distribution $\cal{D}$ returned by the algorithm is obtained by removing at most $28$ directed edges in $E(G)$ and enforcing transitivity. This follows immediately from our choice of $z_1,z_2$ and $z_3$. Next consider any $u,v \in V(G)$ and any $Q \in \supp(\cal{D})$ such that $(u,v) \not\in Q$. Let $S \subseteq E(G)$ denote the set of edges that we remove before enforcing transitivity to obtain $Q$. Now  we have $|S| \leq 28$. We also have that there exist $e_1,e_2 \in S$ such that removing $e_1$ and $e_2$ from $E(G)$ and enforcing transitivity ensures that $(u,v)$ is not in the resulting quasipartition. This is because $G$ is a directed cycle and there are exactly two directed paths from $u$ to $v$. Therefore it suffices to remove a directed edge from each of the two paths to obtain a quasipartition that does not contain $(u,v)$. This implies that if we consider all possible ways to select two edges (one in the clockwise direction and one in the counterclockwise direction such that they don't overlap i.e. we do not pick $(x,y)$ and $(y,x)$ for all $(x,y) \in E(G)$) from $S$, and then consider the quasipartition obtained by removing them from $E(G)$ and enforcing transitivity we have a set $P_Q$ of at most $\binom{28}{2}$ possible quasipartitions and $(u,v)$ is not in at least one of them. Now we consider the following distribution of quasipartitions. First we pick $Q$ from $\cal{D}$. Then we pick uniformly at random a quasipartition from $P_Q$. This gives us a new distribution of quasipartitions $\cal{D}'$. Note that every quasipartition $Q' \in \supp(\cal{D}')$ is a directed cut metric. This is because removing two directed edges in opposite directions as described earlier and enforcing transitivity on the remaining edges to obtain $Q'$ partitions $V(G)$ into $U \subset V(G)$ and $V(G) \setminus U$ such that $(x,y) \not\in Q'$ if $x \in U$ and $y \in V(G) \setminus U$ and $(x,y) \in Q'$ otherwise. Therefore $\cal{D}'$ is a convex combination of directed cut metrics. Since the bound from lemma \ref{lem:lower prob} decreases by at most a factor of 28 when applied to $\cal{D}'$ it follows that the distortion of $\cal{D}'$ is at most $28$ times larger than that of $\cal{D}$. 
\end{proof}

\section{Embedding directed trees into directed $\ell_1$}\label{sec:treel1}

In this section we describe a method for embedding directed trees into directed $\ell_1$ with distortion one.
Let $G=(V,E)$ be a directed tree, and let $w$ be a weight function on the edges of $G$. Let $M=(X,d)$ be the shortest path quasimetric space induced by $G$. The following algorithm gives us a distribution D over quasipartitions of $G$.

\begin{description}

\item{\textbf{Input:}} A directed tree $G$.
\item{\textbf{Output:}} A random quasipartition $Q$. \\
\item{\textbf{Initialization:}}: 
Set $Q=E(G)$.

\item{\textbf{Step 1.}} Let $W = \sum\limits_{e \in E} w(e)$.
\item{\textbf{Step 2.}} Let $D_E$ be a distribution over $E$, where each edge $e \in E$ is sampled with probability $w(e) / W$.
\item{\textbf{Step 3.}} Pick an edge $e \in E$ from the above distribution, and remove $e$ from $Q$.
\item{\textbf{Step 4.}} 
 Enforce transitivity on $Q$; that is, for all $u,v,w \in V(G)$ if $(u,v) \in Q$ and $(v,w) \in Q$ then add $(u,w)$ to $Q$.

\end{description}

\begin{lemma}\label{lem:treeL_1}
For every $u,v \in V$ we have $\Pr_{Q \sim D}[(u,v) \notin Q]= \frac{d(u,v)}{W}$.
\end{lemma}

\begin{proof}
Let $P=(a_1=u, a_2, \ldots, a_m=v)$ be the unique shortest path from $u$ to $v$ in $G$. We have that $(u,v) \notin Q$ iff there exists an $i \in \{1,2,\ldots, m-1\}$ such that $(a_i,a_{i+1}) \notin Q$. By the construction, for every $i \in \{1,2,\ldots, m-1\}$ we have that $\Pr_{Q \sim D}[(a_i,a_{i+1}) \notin Q]= \frac{d(a_i,a_{i+1})}{W}$. Now note that the algorithm only removes one edge from $Q$, and thus we have
\[
\Pr_{Q \sim D}[(u,v) \notin Q]= \sum_{i=1}^{m-1}\Pr_{Q \sim D}[(a_i,a_{i+1}) \notin Q] = \frac{d(u,v)}{W},
\]
as desired.
\end{proof}

Now we are ready to prove the main result of this section.

\begin{proof}[Proof of Theorem \ref{thm:treel1}]
The above algorithm gives us a convex combination of quasipartitions, and thus a convex combination of $0$-$1$ quasimetrics with distortion one. The support of this convex combination consists of $0$-$1$ quasimetrics that are also directed cut metrics. This is equivalent to an embedding into directed $\ell_1$, as desired.
\end{proof}


\section{Applications to directed cut problems}\label{sec:application}
\paragraph{Directed \mcut problem} Consider the Directed \mcut problem. Let $P$ be the set of all directed paths from a source terminal to its corresponding sink terminal. For every $e \in E(G)$, we define an indicator variable $x(e)$ that indicates whether $e$ belongs to a cut or not. We have the following integer program for the problem:

\begin{equation*}
\begin{array}{ll@{}ll}
\text{minimize}  & \displaystyle\sum\limits_{e \in E(G)} c(e)&x(e) &\\
\text{subject to}& \displaystyle\sum\limits_{e \in p}   &x(e) \geq 1,  &\forall p \in P\\
                 &                                                &x(e) \in \{0,1\}, &\forall e \in E(G)
\end{array}
\end{equation*}

By relaxing this integer program, we get the following LP relaxation:

\begin{equation*}
\begin{array}{ll@{}ll}
\text{minimize} & \displaystyle\sum\limits_{e \in E(G)} c(e)&x(e) &\\
\text{subject to}& \displaystyle\sum\limits_{e \in p}   &x(e) \geq 1,  &\forall p \in P\\
                 &                                                &x(e) \geq 0, &\forall e \in E(G)
\end{array}
\end{equation*}

Let $x$ be the solution to the above LP relaxation. Let $M$ be the quasimetric space induced by $x$. Suppose there exists a $(1-\varepsilon)$-bounded $\beta$-Lipschitz distribution over quasipartitions of $M$ $\cal{D}$. Then 
we can sample a quasipartition $Q$ of $M$ from $\cal{D}$. Let $S_Q = E(G) \setminus Q$. Since $Q$ is $(1-\varepsilon)$-bounded, we have that $S_Q$ is a valid \mcut solution. We have:
\[
\mathbb{E}(c(S_Q)) = \sum_{e \in S} \mathbb{E}(c(e)) = \sum_{e \in S} c(e) \cdot \Pr(e \notin Q) \leq \sum_{e \in S} c(e) \cdot \beta\frac{x(e)}{1-\varepsilon} \leq \frac{\beta}{1-\varepsilon} \OPT
\]

This means that in expectation $S_Q$ is a $\frac{\beta}{1-\varepsilon}$-approximation for the optimum \mcut solution. Therefore it follows that there exists a quasipartition $Q^* \in \supp(\cal{D})$ such that $c(S_{Q^*}) \leq \frac{\beta}{1-\varepsilon} \OPT$. 

\begin{proof}[Proof of Theorem \ref{thm:multicutresults}]
The proof follows by combining the aforementioned result with the results of Sections \ref{sec:treewidth} and \ref{sec:pathwidth}.
\end{proof}

\paragraph{Directed Non-Bipartite \spcut} The standard LP relaxation for this problem is as follows:

\begin{equation*}
\begin{array}{ll@{}ll}
\text{minimize} & \displaystyle\sum\limits_{e \in E(G)}  c(e)x(e) \\
\text{subject to}& \displaystyle\sum\limits_{(s_i,t_i) \in T} \dem(i) d_x(s_i,t_i)  \geq 1 \\
                 &                                            x(e) \geq 0, &\forall e \in E(G)
\end{array}
\end{equation*}

Here $d_x(u,v)$ denotes the shortest path distance induced by the function $x$ in $G$. We also denote the optimal value of the objective function by $\OPT$. Note that we may assume that $\displaystyle\sum\limits_{(s_i,t_i) \in T} \dem(i) d_x(s_i,t_i)  = 1$ when the optimum is achieved.

It is shown in \cite{klein1993excluded} that the existence of a $(\sigma q, \sigma' \frac{C}{q})$-decomposition for undirected graphs implies that the integrality gap for the (undirected) uniform demand \spcut problem is $O(\sigma \sigma')$. A similar argument to the one used in Lemma 3.5 in that paper gives the following result.

\begin{theorem}\label{thm:sparsestcut}
Let $G$ be a directed graph. Suppose for all $r>0$ and all $w:E(G) \to \{0 \cup \mathbb{R}^{+} \}$ there exists an $r$-bounded $\beta$-Lipschitz distribution over quasipartitions of the shortest path quasimetric induced by $w$ on $G$. Then the integrality gap of the LP relaxation for the Directed Non-Bipartite \spcut problem with uniform demands and any choice of capacities on $G$ is $O(\beta)$.
\end{theorem}

\begin{proof}
It is known that the solution of the LP relaxation for the Directed Non-Bipartite \spcut \newline
problem is a weight function $w:E(G) \to \{0 \cup \mathbb{R}^{+} \}$ on $G$ that induces a quasimetric space $M=(V(G),d)$ \cite{memoli2016quasimetric}. Let $\cal{D}$ be a $\frac{1}{4n^2}$-bounded $\beta$-Lipschitz distribution over $M$. Every quasipartition $Q \in \supp(\cal{D})$ induces a cut $S_Q$ of $G$. If we pick randomly pick $Q$ from $\cal{D}$ we have that the $\displaystyle \mathop{\mathbb{E}}_{Q \sim \cal{D}} [C(S_Q)] \leq \sum\limits_{i} c(e)\frac{\beta d(e)}{r} \leq 4n^2 \beta \OPT$. This implies that there exists a quasipartition $Q^* \in \supp(\cal{D})$ such that $S_{Q^*} \leq 4n^2 \beta \OPT$. Now consider the cut $S_{Q^*}$. Let $T = {U_1,U_2 , \ldots , U_m}$ be the collection of vertex sets of all maximal strongly connected components of $G[E \setminus S_{Q^*}]$. Suppose $|U_i| < \frac{2n^2}{3}$ for all $i \in \{1,\ldots,m\}$. Then there exists $r \in \{1,\ldots,m\}$ such that $\frac{n}{6} \leq |U_1 \cup \ldots \cup U_r| \leq \frac{5n}{6}$ and $\frac{n}{6} \leq  |U_{r+1} \cup \ldots \cup U_m| \leq \frac{5n}{6}$. Furthermore we have that $D(S_{Q^*}) \geq \frac{5n^2}{36}$. This is because for all $u \in \{U_1 \cup \ldots \cup U_r \}$ and $v \in \{U_{r+1} \cup \ldots \cup U_m \}$ we have that either there is no path from $u$ to $v$ or there is no path from $v$ to $u$ in $G[E \setminus S_{Q^*}]$ since each $U_i$ is a maximal strongly connected component. So we have that the sparsity of $ S_{Q^*}]$ is at most $\frac{4 n^2 \beta \OPT}{5n^2/36} = O(\beta \OPT)$. Suppose $|U_i| \geq \frac{2n^2}{3}$ for some $i \in \{1,\ldots,m\}$. Then we have that for all $u,v \in U_i$ $d(u,v) \leq \frac{1}{4n^2}$ and $d(v,u) \leq \frac{1}{4n^2}$. So we can use the same argument from Lemma 16 in  \cite{leighton1999multicommodity} to obtain a cut with sparsity $O(\beta \OPT)$. This implies that the integrality gap of the LP relaxation is $O(\beta)$ concluding the proof.  
\end{proof}

\begin{proof}[Proof of Theorem \ref{thm:uniformgap}]
The proof follows by combining 
Theorem \ref{thm:sparsestcut} with Theorems \ref{thm:tw2} and \ref{thm:boundedpathwidth}.
\end{proof}

\begin{proof}[Proof of Theorem \ref{thm:spcutresults}]
The proof follows by combining 
Theorem \ref{thm:sparsestcut} with Theorems \ref{thm:tw2} and \ref{thm:boundedpathwidthalgo}. 
\end{proof}


\begin{proof}[Proof of Theorem \ref{thm:nonuniformgap}]
Theorem A.1 from \cite{charikar2006directed} due to Charikar, Makarychev and Makarychev implies that, the non-uniform flow-cut gap for a digraph $G$ is upper bounded by the minimum distortion for embedding any shortest path quasimetric space supported on $G$, into a convex combination of 0-1 quasimetrics. Combined with Theorems \ref{thm:treel1} and \ref{thm:cycle_01} this implies the statement of the theorem.
\end{proof}

\section{Limitations of the Klein-Plotkin-Rao partitioning scheme}\label{sec:limits}
Let $G$ be a directed planar graph, and let $M$ denote the shortest-path quasimetric space of $G$. In order to find a constant Lipschitz distribution over $r$-bounded quasipartitions of $M$, one may try to generalize previous algorithm of \cite{klein1993excluded}, where their work is on undirected graphs. The following algorithm is implicit in their work for planar undirected graphs.

\begin{description}
\item{\textbf{Input:}}
A planar connected graph $G$.
\item{\textbf{Output:}}
A random $r$-bounded partition $Q$ of the shortest-path metric space of $G$.
\item{\textbf{Initialization.}}
Set $G^* = G$ and $Q=E(G)$.

\item{\textbf{Step 1.}} 
Pick $z \in [0,r]$ uniformly at random. Let $c=3$.

\item{\textbf{Step 2.}} 
For every connected component $C$ of $G^*$, proceed as follows. Pick an arbitrary $x \in V(C)$. For all $(u,v) \in E(G^*)$ remove $(u,v)$ from $Q$ and $E(G^*)$ if $d_{G^*}(x,v)> i\cdot r + z$ and $d_{G^*}(x,u) \leq i\cdot r + z$ for some integer $i \geq 0$. Set $c := c - 1$.

\item{\textbf{Step 3.}} 
If $c> 0$, recursively call Steps $1-2$ for every connected component of $G^*$.

\item{\textbf{Step 4.}} 
Enforce trasitivity on $Q$. 
\end{description}

\begin{figure}[h]
\begin{center}
\scalebox{1.3}{\includegraphics{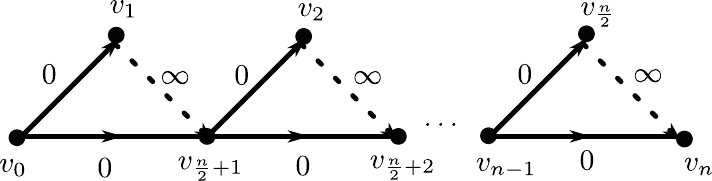}}
\caption{The graph $G$}
\label{fig:kprcounter}
\end{center}
\end{figure}

Note that in this algorithm, we start with $c = 3$, and thus the resulting partition is constant Lipschitz. They show that for planar graphs, $3$ rounds is enough to get an $r$-bounded partition. However, this might not be true in the directed setting.
A natural generalization of the above algorithm to the directed setting is as follows. Instead of undirected edges in step 2, we consider directed ones, and in step $3$, we consider weakly connected components of $G^*$. This generalization fails to result in an $r$-bounded quasipartition for the directed case. We provide the following counter example.
Let $G$ be the graph as shown in figure \ref{fig:kprcounter}. It is easy to check that $G$ is a graph of path-width $2$. Now suppose that the first time we call step $2$ of the algorithm, we pick $x = v_0$. Therefore, no edges will be removed from $Q$. Now suppose that the next call of step $2$, picks $x=v_1$, and the next call picks $x=v_2$, and so on. After $3$ levels of calling the second step, we do not get an $r$-bounded quasipartition. This means that in order to get an $r$-bounded quasipartition, we need to call this step $\Omega(n)$ times, and thus the resulting quasipartition will not be $O(1)$-Lipschitz.

\section{Lower bound for random embeddings of the directed cycle into trees}\label{sec:lowerbound}

\begin{proof}[Proof of Theorem \ref{thm:lowerbound}]
We will consider the directed cycle in the following description. Let the set of vertices $V = \{ v_0,v_2, \ldots, v_{n-1} \}$. Let the set of edges be $E = \{(v_i,v_{(i+1 \bmod n)})| i \in [0,n-1] \}$. Furthermore let all the edges have unit weight.

Let $S$ denote the set of all non-contracting embeddings of $G$ into directed trees. Let $Q$ denote the set of all distributions over edges in $G$. Let $D$ denote the set of all distributions over $S$. The least distortion of any random embedding of $G$ into directed trees is denoted by $\OPT$ where $\OPT = \min\limits_{h \in D} \max\limits_{(u,v) \in E} \displaystyle \mathop{\mathbb{E}}_{H \sim F} {\left[\frac{d_{H}(u,v)}{d_G(u,v)} \right]} = \min\limits_{F \in D} \max\limits_{(u,v) \in E} \displaystyle \mathop{\mathbb{E}}_{H \sim F} {[d_{H}(u,v)]}$ since all edges have unit weight. This is lower bounded by the Von Neumann dual problem given by $\max\limits_{q \in Q} \min\limits_{H \in S} \displaystyle \mathop{\mathbb{E}}_{(u,v) \sim q} \left[d_{H}(u,v)\right]$. Therefore for all $q \in Q$ we have that $\min\limits_{H \in S} \displaystyle \mathop{\mathbb{E}}_{(u,v) \sim q} \left[d_{H}(u,v)\right] \leq \OPT$. Now we will set $q$ to be the uniform distribution over edges in $E$. Next consider any $H \in S$. We have that $\displaystyle \mathop{\mathbb{E}}_{(u,v) \sim q} \left[d_{H}(u,v)\right] = \frac{1}{n} \sum\limits_{i \in [0,n-1]} d_H(v_i,v_{(i+1 \bmod n)})$. Let us denote the expression $\sum\limits_{i \in [0,n-1]} d_H(v_i,v_{(i+1 \bmod n)})$ by $X$. Let $L$ be the set of vertices of degree $1$ in the underlying undirected graph of $H$. Furthermore let the vertices in $L$ be labeled as follows. $L = \{ v_{x_0},v_{x_2}, \ldots, v_{x_{m-1}}\}$ where for any $j \in [0,m-1]$ we have that $x_j \in [0,n-1]$ and $x_1 < x_2 \ldots < x_m$. We have that $X \geq \sum\limits_{j \in [0,m-1]} d_H(v_i,v_{(i+1 \bmod m)}) $. Let us denote $\sum\limits_{j \in [0,m-1]} d_H(v_i,v_{(i+1 \bmod m)}) $ by $Y$. Note that $Y$ is the length of a walk traversing all the leaf vertices in $H$ in the order given by the indices of the vertices in $L$ and terminating at the starting vertex $v_{x_0}$. Since every directed edge in $H$ separates at least one pair of leaves in one direction it must be that every directed edge in $H$ is traversed at least once during this walk. So this implies that $Y \geq \sum\limits_{(u,v) \in E(H)} d_H(u,v)$. Now consider any directed edge $(u,v) \in E(H)$. We have that $d_H(u,v) + d_H(v,u) \geq n$. This is because $H$ is non-contracting and for any $u,v \in V(G)$ we have that $d_G(u,v) + d_G(v,u) = n$. Therefore we have that $X \geq Y \geq n(n-1)$. This implies that given our choice of $q$ and for any $H \in S$ we have $\displaystyle \mathop{\mathbb{E}}_{(u,v) \sim q} \left[d_{H}(u,v)\right] = \frac{1}{n} X \geq \Omega(n)$. Therefore $\OPT \geq \Omega(n)$, which concludes the proof.  
\end{proof}

\bibliography{bibfile}

\end{document}